\documentclass[12pt,draftcls,onecolumn]{IEEEtran}

\usepackage{graphicx}
\usepackage{amsmath,graphicx,float,cite,dsfont,amssymb}
\usepackage{amsmath,amssymb,amsthm}
\usepackage{amsfonts,bm} % Math packages
\usepackage{mathtools}
\usepackage{color}
\usepackage[inline, shortlabels]{enumitem}
\usepackage{pgfplots}
\newtheorem{theorem}{Theorem}

\newtheorem{algorithm}{Algorithm}
\newtheorem{lemma}{Lemma}

\let\emptyset\varnothing
\DeclarePairedDelimiter\normV{\lVert}{\rVert}
\DeclarePairedDelimiter\norm{\lVert}{\rVert}

\newcommand{\defeq}{\vcentcolon=}

\usepackage{textcomp}
\usepackage{algorithm}
\usepackage{algorithmic}
\usepackage{caption}
\usepackage{tikz}
\usepackage{graphicx}
\usepackage{balance}

\usepackage[makeroom]{cancel}
\usepackage{colortbl}
\theoremstyle{remark}

\let\emptyset\varnothing

\DeclareMathOperator{\EX}{\mathbb{E}}% expected value

\theoremstyle{definition}

\usepackage{setspace}
\usepackage{subcaption}
\usepackage[export]{adjustbox}
\usepackage{ragged2e}
\begin{document}

\title{Rate-Splitting Multiple Access in Cache-Aided Cloud-Radio Access Networks}
\author{\IEEEauthorblockN{Robert-Jeron Reifert, \IEEEmembership{Student Member, IEEE}, Alaa Alameer Ahmad, \IEEEmembership{Member, IEEE}, Yijie Mao, \IEEEmembership{Member, IEEE}, Aydin Sezgin, \IEEEmembership{Senior Member, IEEE}, and Bruno Clerckx, \IEEEmembership{Senior Member, IEEE}\\}	
\thanks{This work has been partially funded by the Federal Ministry of Education and Research (BMBF) of the Federal Republic of Germany (F\"orderkennzeichen 01IS18063A, ReMiX) and by the U.K. Engineering and Physical Sciences Research Council (EPSRC) under grant EP/N015312/1, EP/R511547/1. \newline
\IEEEauthorblockA{\indent Robert-Jeron Reifert, Alaa Alameer Ahmad and Aydin Sezgin are with Digital Communication Systems, Ruhr-Universit\"at Bochum, Bochum, Germany. (Email: \{robert-.reifert,alaa.alameerahmad,aydin.sezgin\}@rub.de)}\newline
\indent Yijie Mao and Bruno Clerckx are with with the Communications and Signal Processing Group, Department of Electrical and Electronic Engineering, Imperial College London, U.K. (Email: \{y.mao16, b.clerckx\}@imperial.ac.uk)
}
}
\maketitle
%\vspace{-6cm}
\begin{abstract}
Rate-splitting multiple access (RSMA) has been recognized as a promising physical layer strategy for 6G. Motivated by ever increasing popularity of cache-enabled content delivery in wireless communications, this paper proposes an innovative multigroup multicast transmission scheme based on RSMA for cache-aided cloud-radio access networks (C-RAN). Our proposed scheme not only exploits the properties of content-centric communications and local caching at the base stations (BSs), but also incorporates RSMA to better manage interference in multigroup multicast transmission with statistical channel state information (CSI) known at the central processor (CP) and the BSs. At the RSMA-enabled cloud CP, the message of each multicast group is split into a private and a common part with the former private part being decoded by all users in the respective group and the latter common part being decoded by multiple users from other multicast groups. Common message decoding is done for the purpose of mitigating the interference. In this work, we jointly optimize the clustering of BSs and the precoding with the aim of maximizing the minimum rate among all multicast groups to guarantee fairness serving all groups. The problem is a mixed-integer non-linear stochastic program (MINLSP), which is solved by a practical algorithm we proposed including a heuristic clustering algorithm for assigning a set of BSs to serve each user followed by an efficient iterative algorithm that combines the sample average approximation (SAA) and weighted minimum mean square error (WMMSE) to solve the stochastic non-convex sub-problem of precoder design. Numerical results show the explicit max-min rate gain of our proposed transmission scheme compared to the state-of-the-art trivial interference processing methods. Therefore, we conclude that RSMA is a promising technique for cache-aided C-RAN.
\end{abstract}
%\vspace{-2cm}
%\begin{keywords}
%Cloud association, iterative auction, knapsack algorithm, distributed implementation.
%\end{keywords}

\IEEEpeerreviewmaketitle

\thispagestyle{empty}

%\vspace{-.45cm}
\section{Introduction}
\subsection{Overview}
The dramatically increasing popularity of smart communications devices has driven the current fifth generation (5G) and future beyond 5G (B5G) of wireless communication networks to handle a tremendous volume of mobile data-traffics \cite{emr}. Recently, wireless caching and physical layer multicast transmission has received considerable attention in the research community as two essential techniques to address the associated data-traffic challenges. Wireless caching brings the content closer to users, and hence helps reducing latency of the requested content delivery and avoiding network congestion. Moreover, multicast transmission in which the same content is transmitted to multiple users offers an ideal transmit platform that fits the feature of content-oriented services. Hence, in Video-on-Demand applications such as Youtube and Netflix, the same content is requested by multiple users which leads to a \textit{content-reuse} request type \cite{DBLP,8782879,6495773}. In its most general form, multigroup multicast transmission, enables simultaneous transmission of multiple distinct messages to several multicast groups of users \cite{4443878}, which however, introduces inter-group interference. Conventionally, the interference is addressed by transmit beamforming design based on treating interference as noise (TIN) strategy at each user \cite{7925639,7488289,6123786,5074583}. However, in general, such strategy is sup-optimal, especially in moderate to strong interference scenarios.

Recently, rate-splitting multiple access (RSMA)-based transmission strategy in which the transmitter splits the message of each user has been widely studied in multi-antenna broadcast channel \cite{9158344} and applied in multigroup multicast transmission \cite{8019852} and cloud-radio access network (C-RAN) \cite{9217249}. %Refer to papers from https://sites.google.com/view/ieee-comsoc-wtc-sig-rsma/best-readings?authuser=0
Specifically, in multigroup multicast transmission, by splitting the group-specific messages into two parts and allow one part (which is known as a common part) to be decoded by multiple users in different multicast groups, RSMA enables partially decoding the inter-group interference and partially treating the inter-group interference as noise. It therefore better manages the inter-group interference. 
The concept of RSMA dates back to the late 1970's when it was first suggested by \cite{Carleial78}. The scheme is shown to attain the best achievable rate region for the two-users interference channel (IC) in \cite{1056307} and later in \cite{4675741}, the authors show that such strategy can achieve to within one-bit of the capacity of two-users IC. Note that recent works mostly assume instantaneous (perfect or imperfect) channel state information (CSI) at the transmitter (CSIT). Contrary to that, in this paper, motivated by the appealing performance of RSMA in statistical CSIT in \cite{ahmad2020rate,yin2021ratesplitting}, we focus on statistical CSIT which is more practical as it requires little communication overhead to acquire.  %The assumption of full CSIT is impractical, especially in dense networks, as an overwhelming amount of CSIs have to be obtained. This would require a huge computational overhead and waste of resources, which excludes the perfect CSIT scenario from practical implementations. Therefore, herein we consider a more practice relevant assumption of statistical CSIT.} \\

This paper considers the problem of optimizing RSMA-based transmission strategies in a C-RAN and makes use of multicast and caching techniques for efficient content delivery. C-RAN is a promising network architecture which provides an ideal platform to enable the B5G wireless systems handling the challenges of mobile communication service applications. 
Under such an architecture, a central processor (CP) located at the cloud employs most but not all baseband processing functions, e.g., encoding of the private and common messages of RSMA, which allows joint optimization of the resources and enables full cooperation amongst the set of base stations (BSs) controlled by the CP. The BSs that are distributed throughout the network are connected to the CP through high-speed digital fronthaul links with limited capacity. 

A vital part of the system model is the location of different parts of baseband processing and the use of the fronthaul link, this is referred to as functional split. Functional splits in C-RAN can be divided into two major categories, referred to as \emph{compression} and \emph{data-sharing} strategies \cite{97813165296}. Under each strategy, there are different usages of the fronthaul link and the baseband processing tasks are allocated in a different way. Utilizing the \emph{compression} strategy, the CP performs every processing task leaving only the radio transmission to the respective BSs. To be more specific, the CP designs the beamforming vectors and performs compression and quantization of the signals to be transmitted. The fronthaul links are then used to forward the signals to the BSs, which thereafter only need to perform transmission. However, the quantization process, necessitated by the limited capacity of the fronthaul links, generates quantization noise diminishing the system performance \cite{6588350}.
As opposed to this, using \emph{data-sharing}, the CP performs channel encoding and the BSs then continue with the remaining baseband processing tasks \cite{DaiY14}. Therefore, under \emph{data-sharing} strategy, the BSs have more computational resources and responsibilities. The computational load is balanced between centralized computations at the CP and local computation at the BSs.
Apart from that, other functional splits between the CP and the BSs can be employed thanks to the flexibility of the C-RAN system model. To this end, in this paper, we utilize the \emph{data-sharing} strategy. % to prevent alleviating the performance due to quantization noise. 
The authors in \cite{7809154} have shown that utilizing the \emph{data-sharing} strategy results in better performance compared to using the \emph{compression} strategy, especially when the fronthaul capacity is limited. \\
\indent Recently, assisting C-RAN with RSMA-based transmission strategies has shown considerable gain in the performance compared to conventional TIN \cite{9145076,8732995,ahmad2020rate,9145363}. Next, we discuss the related works and their relation to our setup.
\subsection{Related Work}
Wireless caching has been studied in many research works as an efficient content-based communication. Motivated by the extreme popularity of video streaming applications, the seminal work \cite{6600983} has shown that small cell dense networks can improve the spectral efficiency of the wireless system. The work \cite{8732315} studies a cache-enabled broadcast-relay wireless network from a latency-centric perspective. Furthermore, the authors in \cite{6600983} proposed to equip the BSs with local memory to cache the most popular content in order to alleviate the bottleneck of fronthaul/backhaul communications and efficiently provide the video content to users. Such a network architecture is best described by C-RAN model. Cache-assisted wireless networks have been intensively explored recently. In \cite{7499119} the authors show that coded caching can significantly reduce the fronthaul and transmit power costs. The authors in \cite{7925639} proposed an efficient algorithm to solve the combinatorial optimization problem of content delivery in a cache-assisted C-RAN. In \cite{8269405} the authors investigate a trade-off cache strategy to jointly minimize the outage probability and fronthaul costs in C-RAN. In the research work \cite{7488289} the authors have shown the benefits of integrating  multicast transmission with local caching in reducing the network-wide fronthaul and transmit power costs in C-RAN. \\
\indent Multicast transmission is essential in content-based applications such as video streaming and video on demand \cite{7537173}, but also in other application areas such as satellite communications \cite{7091022,7765141}. Yet, the seminal work \cite{1634819} has shown that the inherent beamforming vector optimization, even for a single multicast group, is an NP-hard problem. Thus, with multigroup multicast we need to account for inter-group interference which represents the bottleneck in achieving a good performance in the system \cite{8019852}. \\ \indent
Interference is one major research challenge in wireless communication. %The capacity of interference channel is an open problem in information theory. 
Nevertheless, the authors in \cite{4675741} have shown that a RSMA strategy as proposed in \cite{Carleial78, 1056307} can achieve within one-bit of the capacity of two-users IC. Motivated by the theoretical works, the authors in \cite{5910112} show that RSMA transmission scheme can significantly reduce the transmit power costs in multi-cell multi-users communication networks. Recently, several works illustrated the benefits of RSMA in flexibly managing the interference in downlink multiple input single output broadcast channels (MISO-BC) \cite{7513415,7434643,7555358,7867076,8907421,7470942, Mao4 ,Mao,9120024,7972900,mao2019dirty}. In \cite{Mao}, it is shown that RSMA generalizes and outperforms other downlink transmission schemes such as TIN and non-orthogonal multiple access (NOMA). Moreover, information theoretical studies have shown that RSMA achieves the optimal degrees of freedom (DoF) in MISO-BC with imperfect CSIT \cite{7555358, 7972900, mao2019dirty}, which therefore, making all other strategies achieve only the suboptimal DoF in the same setting. The works \cite{8007039,8703386} investigated RSMA cache-aided MISO-BC. In \cite{8007039}, the authors proposed a RSMA scheme combining coded-caching and spatial multiplexing in an overloaded MISO-BC which yields gains in terms of delivery time and CSIT requirements compared to state-of-the-art schemes. The work \cite{8703386} characterized the generalized DoF of the symmetric cache-aided MISO-BC under partial CSIT up to a constant multiplicative factor. The setup of multigroup multicast with RSMA in downlink MISO-BC has been studied in \cite{8019852,8445878}. The authors have investigated the problem of maximizing the minimum rate to guarantee fairness among all multicast groups. RSMA with imperfect CSIT and the goal to maximize the minimum achievable rate among all multicast groups was studied in \cite{9257433} for multigroup multicast and multibeam sattelite systems. \\ 
\indent Interestingly, MISO-BC represents a special case of C-RAN where the capacity of fronthaul links tends to infinity \cite{97813165296}. RSMA assisted C-RAN has been recently studied in several works. In \cite{8732995} the authors show significant gain in terms of spectral efficiency in a downlink C-RAN system where the CP has full knowledge of CSIT. In \cite{ahmad2020rate} RSMA has shown to be robust against imperfections of CSI. In \cite{8756076} the authors apply RSMA technique in a C-RAN where the CP uses a compression strategy to transfer information from the CP to the set of BSs. Other benefits of RSMA for enhancing the performance of downlink C-RAN have been shown in terms of minimizing the transmit power cost \cite{9145363} and maximizing the energy efficiency of the system \cite{9145076}. 
Departing from all previous works, this paper studies an RSMA-assisted downlink C-RAN where the BSs are equipped with local cache memory. The CP assorts to a multigroup multicast transmission scheme for efficient content-delivery and aims at maximizing the minimum rate (a.k.a. max-min fairness (MMF) rate) of all multicast groups. In the next subsection we discuss the contributions made in this paper in more details.
\subsection{Contributions}
In contrast to the recent related literature, to the best of the authors knowledge, this is the first work that considers an RSMA-assisted cache-enabled C-RAN. We focus on the multigroup multicast transmission and consider the maximization of the minimum rate of all multicast groups in the network. The major contributions in this work are summarized as follows:
\begin{itemize}
\item[1)] \textit{Comprehensive System Model:} This paper considers a novel C-RAN system model that integrates caching, RSMA, and multigroup multicast transmission. The functional split between CP and BSs is based on a \emph{data-sharing} strategy. We enable caching at the BSs as they are equipped with local memories to store the popular content closer to the end users. This helps alleviating the traffic on the fronthaul capacity, as BSs that cache a requested file are able to process the data locally. Multigroup multicasting is considered to account for the properties of content-based transmissions by grouping all users requesting the same content in the same multicast group. To mitigate the inter-group interference that limits the performance, we consider utilizing RSMA. In RSMA, each group-specific message is split and encoded into two streams, namely, a private stream that is decoded by all users in the multicast group and treated as noise by all users within other multicast groups, and a common stream that is decoded by multiple users in different multicast groups. Hence, the users are able to not only decode their respective private messages, but also decode a subset of common messages dedicated to other users so as to reduce the interference throughout the network.
\item[2)] \textit{Problem Formulation:} Herein we formulate a resource allocation problem to maximize the minimum achievable rate guaranteeing fairness among all multicast groups in the cache-enabled RSMA-assisted C-RAN. Additionally, we consider that the CP only knows the channel statistics, i.e., we consider statistical CSIT, in contrast to the full CSIT assumption adopted by most works in C-RAN literature. The considered problem is a mixed-integer non-linear stochastic program (MINLSP) which is known to be NP-hard.
\item[3)] \textit{Proposed Algorithms:} To tackle such a challenging problem, we propose an optimization framework consisting of a novel multicast group-based clustering algorithm, the sample average approximation (SAA), and a weighted minimum mean square error (WMMSE)-based algorithm. 
\item[4)] \textit{Numerical Simulations:} Through extensive numerical simulations we evaluate the performance of our proposed scheme. To benchmark the performance, we compare our proposed RSMA scheme with classical TIN and a single common message RSMA (SCM-RSMA) scheme, which defines a single common message that is decoded by all users. Results imply that the performance gain of the proposed RSMA scheme over the benchmarks increases with fonthaul capacity and also with the number of BSs. The gain is most significant in low fronhaul regime with small cache sizes, as another result shows. Moreover, the MMF-rate of the proposed RSMA scheme increases with the utilization of bigger caches and also more transmit antennas. With increasing number of users, we observe a gain of the multigroup multicast transmission over using a simpler transmission scheme. At last, we highlight the numerical features of our proposed method.

\end{itemize}   
\subsection{Notations \& Organization}
%\begin{table}[b]
%	% increase table row spacing, adjust to taste
%	%		\renewcommand{\arraystretch}{1.3}
%	\caption{List of Notations}\vspace*{-.3cm}
%	\label{tb:notations}
%	\centering
%	\newcolumntype{L}[1]{>{\raggedright\arraybackslash}p{#1}}
%	\begin{tabular}{|L{0.09\textwidth}|L{0.5\textwidth}|}
%		\hline
%		Notation & Definition\\
%		\hline
%		\hline
%		$\mathbb{R}$ & Real field\\
%		$\mathbb{C}$ & Complex field\\
%		$\mathbf{x}$ & Bold-face, lower-case: vector\\
%		$\mathbf{X}$ & Bold-face, upper-case: matrix\\
%		$\mathcal{X}$ & Calligraphic: set\\
%		$|\mathcal{X}|$ & Cardinality of set $\mathcal{X}$\\
%		$\mathbf{0}_N$ & Zero vector of length $N$\\
%		$\mathbf{I}_N$ & Identity matrix of size $N\times N$\\
%		$\mathbb{E}\{\cdot\}$ & Expectation\\
%		$ {\rm vec}(\mathcal{X}) $ & Column vector holding all elements of set $\mathcal{X}$ \\
%		 & If $\mathcal{X}=\{x_1,\ldots,x_N\}$, $ {\rm vec}(\mathcal{X}) \equiv [x_1,\ldots,x_N]^T $\\
%		 & If $\mathcal{X}=\{\mathbf{x}_1,\ldots,\mathbf{x}_N\}$, $ {\rm vec}(\mathcal{X}) \equiv [\mathbf{x}_1^T,\ldots,\mathbf{x}_N^T]^T $\\
%		 $\left|\cdot\right|$ & Absolute value\\
%		 $\normV[]{\cdot}_{p}$ & $l_p$-norm\\
%		 $\left( \cdot \right)^{T}$ & Transpose operator\\
%		 $\left( \cdot \right)^{H}$ & Hermitian transpose operator\\
%		\hline
%	\end{tabular}
%\end{table}%
\begin{table}[b]
	% increase table row spacing, adjust to taste
	%		\renewcommand{\arraystretch}{1.3}
	\caption{List of Mathematical Notations}\vspace*{-.3cm}
	\label{tb:notations2}
	\centering
	\newcolumntype{L}[1]{>{\raggedright\arraybackslash}p{#1}}
	\begin{tabular}{|L{0.10\textwidth}|L{0.7\textwidth}|}
		\hline
		Notation & Definition\\
		\hline
		\hline
%		$N$ & Number of BSs\\
%		$K$ & Number of users\\
%		$F$ & Number of files in the library\\
%		$B$ & Transmission bandwidth\\
%		$L$ & Number of antennas per BS\\
%		$G$ & Number of multicast groups\\
%		$C_n^{\text{max}}$ & Fronthaul capacity per BS\\
%		$P_n^{\text{max}}$ & Maximum transmit power per BS\\
		$\mathcal{N}$ & Set of BSs\\
		$\mathcal{K}$ & Set of users\\
		$\mathcal{G}$ & Set of multicast groups\\
		$\mathcal{G}_g$ & Set of users in multicast group $g$\\
		$\mathcal{F}$ & Library of files\\
		$\mathcal{G}_n^p$,$\mathcal{G}_n^c$ & Subset of BSs serving the private/common stream of group $g$\\
		\hline
%		$\mathcal{G}_n^c$ & Subset of BSs serving the common stream multicast group $g$\\
%		$\mathbf{h}_{n,k}$ & Channel vector between BS $n$ and user $k$\\
%		$\mathbf{h}_{k}$ & Aggregate channel vector of user $k$\\
%		$v_g$ & Data chunk of the file requested by group $g$ split into $v_g^p$ and $v_g^c$\\
%		$s_g^p$ & Encoded private stream of group $g$\\
%		$s_g^c$ & Encoded common stream of group $g$\\
%		$\mathbf{w}_{n,g}^p$ & Beamforming vector of BS $n$ to serve the private stream of group $g$\\
%		$\mathbf{w}_{n,g}^c$ & Beamforming vector of BS $n$ to serve the common stream of group $g$\\
%		$\mathbf{w}_{g}^p$ & Aggregate beamforming vector of the private stream of group $g$\\
%		$\mathbf{w}_{g}^c$ & Aggregate beamforming vector of the common stream of group $g$\\
%		$R_g^c$ & Instantaneous achievable rate of the private stream of group $g$\\
%		$R_g^c$ & Instantaneous achievable rate of the common stream of group $g$\\
%		$\bar{R}_g^c$ & Ergodic rate of the private stream of group $g$\\
%		$\bar{R}_g^c$ & Ergodic rate of the common stream of group $g$\\
%		$f_g$ & File requested by multicast group $g$\\
%		$c_{f_g,n}$ & Binary number indicating whether $f_g$ is cached at BS $n$ or not\\
		$\mathcal{M}_g$ & Set of users decoding $s_g^c$\\
		$\Phi_{k}$ & Set of groups whose common messages are decoded by user $k$\\
		$\pi_{k}$ & Decoding order at user $k$\\
		& $\pi_{k}(g_1)>\pi_{k}(g_2)$: user $k$ decodes $g_2$'s common message first\\
%		$y_k$ & Received signal at user $k$\\
%		$T_{g,k}^p$ & Average power of messages received at user $k$ when decoding $s_g^p$\\
%		$T_{i,k}^p$ & Average power of messages received at user $k$ when decoding $s_i^c$\\
%		$I_{g,k}^p$ & Interference-plus-noise at user $k$ when decoding $s_g^p$\\
%		$I_{i,k}^p$ & Interference-plus-noise at user $k$ when decoding $s_i^c$\\
		$\Psi_{i,k}$ & Set of groups whose common messages are decoded by user $k$ after decoding $s_i^c$\\
		$\tilde{\Psi}_{i,k}$ & Set of groups whose common messages are decoded by user $k$ before decoding $s_i^c$\\
		\hline
%		$s_g^p$,$s_g^c$ & Encoded private/common stream of group $g$\\
%		$\gamma_{g,k}^p$,$\gamma_{i,k}^c$ & \textcolor{red}{SINR} of user $k$ when decoding $s_g^p$/$s_i^c$\\
%		$\mathbf{h}_{n,k}$ & Channel vector between BS $n$ and user $k$\\
%		$\mathbf{h}_{k}$ & Aggregate channel vector of user $k$\\
%		$\mathbf{w}_{n,g}^p$,$\mathbf{w}_{n,g}^c$ & Beamforming vector of BS $n$ to serve $s_g^p$/$s_i^c$\\
%		$\mathbf{w}_{g}^p$,$\mathbf{w}_{g}^c$ & Aggregate beamforming vector to serve $s_g^p$/$s_i^c$\\
%		$R_g^p$,$R_g^c$ & Instantaneous achievable rate of $s_g^p$/$s_i^c$\\
%		$\bar{R}_g^p$,$\bar{R}_g^c$ & Ergodic rate of $s_g^p$/$s_i^c$\\
%		$\gamma_{i,k}^c$ & SINR of user $k$ when decoding $s_i^c$\\
%		\hline
	\end{tabular}
\end{table}%
As for the notations of this paper, we denote vectors, matrices, and sets as boldface lower-case, boldface-capital, and calligraphic letters, respectively (e.g., $\mathbf{x}$, $\mathbf{X}$, $\mathcal{X}$). The cardinality of a set $\mathcal{X}$ is given by $|\mathcal{X}|$. A vectorization operator $ {\rm vec}(\cdot) $ is used throughout the paper. To be more specific, $ {\rm vec}(\mathcal{X}) $ is a column vector holding all elements of the set $\mathcal{X}$. It is defined as $ {\rm vec}(\mathcal{X}) \equiv [x_1,\cdots,x_N]^T $ or $ {\rm vec}(\mathcal{X}) \equiv [\mathbf{x}_1^T,\cdots,\mathbf{x}_N^T]^T $ when the elements in $\mathcal{X}$ are scalars or vectors, respectively. Also, $\mathbb{R}$ and $\mathbb{C}$ denote the real and complex field and the expectation of a random variable writes $\mathbb{E}\{\cdot\}$. The transpose and hermitian transpose operators are denoted as $\left( \cdot \right)^{T}$ and $\left( \cdot \right)^{H}$, respectively. Finally, $\left|\cdot\right|$ is the absolute value and $\normV[]{\cdot}_{p}$ is the $l_p$-norm. \\
%Please refer to Table~\ref{tb:notations} for a complete list of notations used throughout this work. \textcolor{red}{This is not the notation table I mean. What I mean is specify the exact notations used in the paper: i.e., $\mathcal{F}$ is the library of files. $\mathcal{G}$ the set of multicast groups; $\mathcal{G}_g$ the set of user in group-$g$ ..... specify all mathematical notations.} 
\indent Please refer to Table~\ref{tb:notations2} for a list of mathematical notations used throughout this work. We pursue the following structural organization in this paper: In Section~\ref{sec:SM} we present our system model consisting of signal, channel, cache, and receiver models. The considered optimization problem is formulated in Section~\ref{sec:PF} followed by the optimization algorithms we proposed for solving the problem including a clustering algorithm to determine the BS clustering, as well as an SAA and WMMSE-based algorithm to optimize the precoders. To evaluate the performance of the proposed algorithm, we conduct numerical simulations in Section~\ref{sec:NS}. Finally, Section~\ref{sec:C} concludes the paper.
\section{System Model}\label{sec:SM}

\begin{figure}
	\begin{center}
		\includegraphics[width=0.7\linewidth]{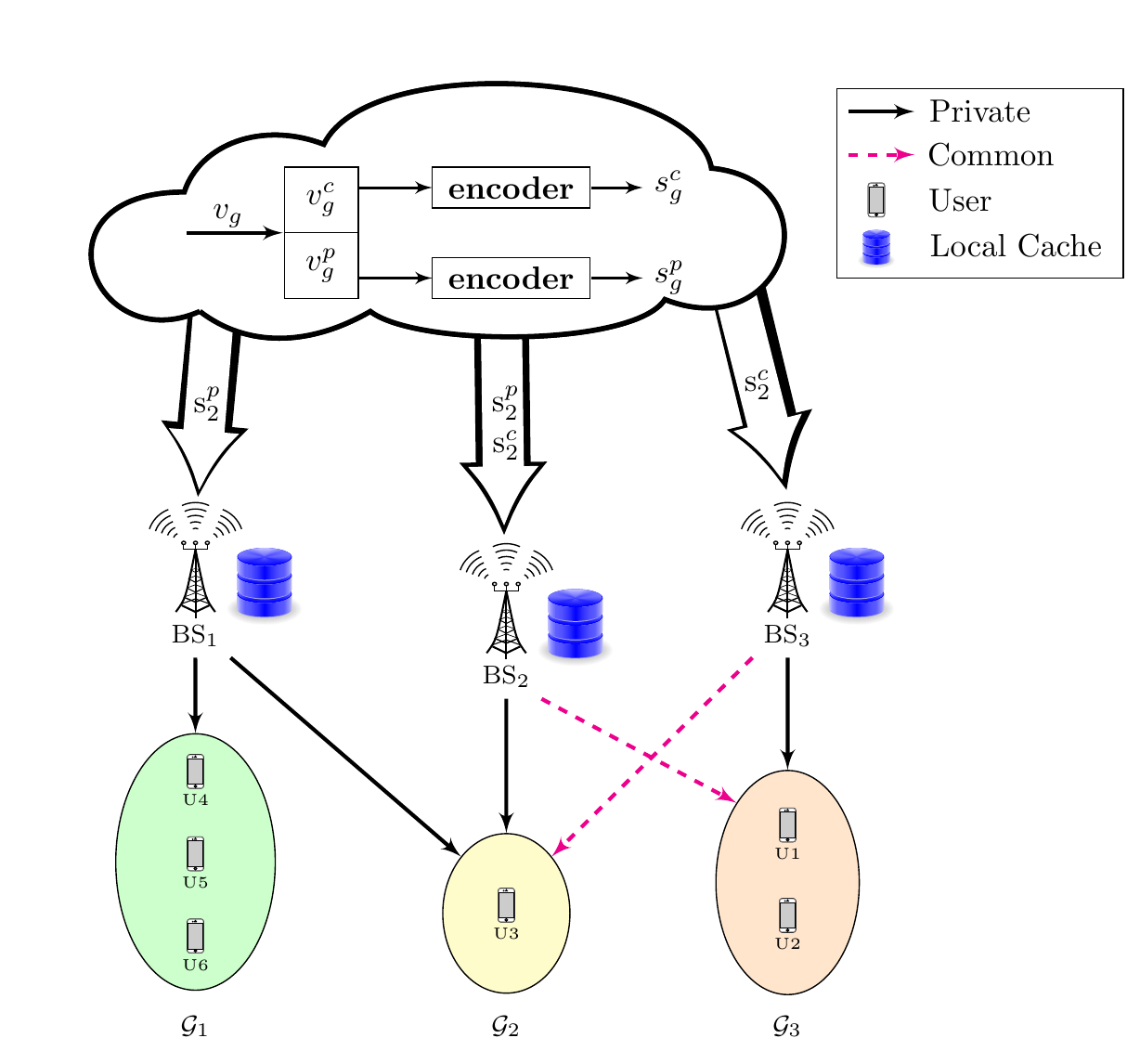}
		\caption{A C-RAN system with three BSs serving three groups of users with multicast messages. Private and common messages are encoded at the cloud.  As a simple illustration we show the shared private and common streams of group $\mathcal{G}_2$}
			\label{Fig1}
	\end{center}

   \vspace{-0.5cm}
\end{figure}
We consider a C-RAN system operating in downlink mode with a transmission bandwidth $B$. The network consists of a set of multiple-antenna BSs, $\mathcal{N} = \left\lbrace1,2,\ldots,N\right\rbrace$, serving a group of single-antenna users indexed by $\mathcal{K} = \left\lbrace1,2,\ldots,K\right\rbrace$. Each BS $n \in \mathcal{N}$ is equipped with $L > 1$ antennas and a local cache memory of finite size. BS $n$ is connected to a CP, located at the cloud, via a fronthaul link of capacity $C_n^{\text{max}}$. The CP has access to the library of files containing the set $\mathcal{F} = \left\lbrace1,2,\ldots,F\right\rbrace$, where the total number of files is $F$. At the beginning of each transmission block, each user in the network submits a request to receive a specific file from the library at the cloud. We assume that each requested content can be fetched and processed by the CP before sharing it with a cluster of BSs for the respective users. The users submit their demands according to a specific demand probability and we assume the probability distribution of the users demands are available at the CP. All users requesting the same content are grouped together and served with a multicast transmission from a cluster of BSs as shown in Fig.~\ref{Fig1}. 
Without loss of generality, we consider that the CP divides each requested file into several data chunks, so that the transmission of each file may take place on several consecutive transmission blocks and the number of required transmission blocks to transmit each file may be different from other files. Let the total number of multicast groups be $G$ such that $1 \leq G\leq \text{min}\{F,K\}$. The set of all multicast groups is given by $\mathcal{G} = \left\lbrace 1,\ldots,G\right\rbrace$ and the set of users in each group $g$ is denoted as $\mathcal{G}_g$. We assume that each user in the network can submit a request to one content at a time, such that $\mathcal{G}_i\cap \mathcal{G}_j = \emptyset,  i\neq j $ and $\sum_{g=1}^{G}|\mathcal{G}_g| \leq  K$. Let $v_g(t)$ be the data chunk of the file requested by group $g$ at time slot $t$.
 \subsection{Received Signal Model}
 In downlink C-RAN, the data chunks in time-slot $t$ are encoded at the CP into streams $s_g(t)$, $\forall g\in \mathcal{G}$.
 In \emph{data-sharing}, due to limited capacity of the fronthaul links, the CP shares private and common parts of $s_g(t)$ with a subset of BSs that are in the serving cluster of the multicast group $g$ and that do not cache the content locally. However,  the file is processed locally at the BS $n$ if it caches the requested content and participates in transmitting it to the intended multicast group. Hence, by using local caches the system can balance the traffic load on fronthaul links. Upon receiving these signals, the selected cluster of BSs cooperatively transmit the coded streams by cooperative beamforming. Hence, BS $n$ constructs $\mathbf{x}_{n} \in \mathbb{C}^{L\times1}$, and sends it according to the following transmit power constraint
 \begin{equation}\label{eq:e1}
 \EX\left\lbrace \mathbf{x}_{n}^{H}(t)\mathbf{x}_{n}(t)\right\rbrace \leq P_{n}^{\text{max}} \quad \forall n \in \mathcal{N},
 \end{equation}
 where $P_{n}^{\text{max}}$ is the maximum transmit power available at BS $n$.
 Let $\mathbf{h}_{n, k}(t) \in \mathbb{C}^{L\times1}$ denote the channel vector between BS $n$ and user $k$, and $\mathbf{h}_k(t) = {\rm vec}\left(\left\{\mathbf{h}_{n, k}(t)|\,\, n \in \mathcal{N}\right\}\right)  \in \mathbb{C}^{NL \times 1}$ be the aggregate channel vector of user $k$. We can write the received signal at user $k$ as 
 \begin{equation}\label{eq:e2}
 y_k(t) = \mathbf{h}_{k}^H(t) \mathbf{x}(t) + n_k(t),
 \end{equation}
 where $n_k(t) \sim \mathcal{CN}\left(0, \sigma_k^2\right)$ is the additive white Gaussian noise (AWGN), and we have $\mathbf{x}(t) = {\rm vec}\left(\left\{ \mathbf{x}_{n}(t)| \,\, n \in \mathcal{N} \right\}\right)$. Without loss of generality, we assume the noise power to be the same for all users, i.e., $\sigma_k^2 = \sigma^2, \forall k \in \mathcal{K}$. 
 \subsection{Channel Fading Model and CSI Uncertainty}
Let us define the instantaneous channel state at time slot $t$ as $\mathbf{h}(t) \triangleq {\rm vec}\left(\left\lbrace\mathbf{h}_{k}(t) |\,\, \forall k\in\mathcal{K} \right\rbrace\right)  \in \mathbb{C}^{NLK\times 1}$. This paper considers a block-fading model in which the channel state $\mathbf{h}(t)$ remains constant over multiple time slots and may vary independently in a random fashion from one block to another according to some stochastic process. Specifically, in block $b$ with length $t_b$, the following relation in the block-fading model is satisfied
\begin{equation}
\mathbf{h}(t) = \mathbf{h}(b), \forall t \in \left\lbrace(b-1)t_b+1,\ldots, b t_b  \right\rbrace.
\end{equation}
Throughout the paper, we focus on optimizing the transmission scheme's parameters, e.g., beamforming vectors, and the resource allocation strategy using the available CSI. Hence, we next drop the dependency on the time variable $t$ for the brevity of notation and focus on the channel state in one transmission block. %The channel estimate at the CP \textcolor{red}{What should come here?}\\

To this end, we assume that the channel between BS $n$ and user $k$ follows the distribution $\mathbf{h}_{n, k} \sim \mathbf{\mathcal{CN}}\left(0, \mathbf{Q}_{n, k}\right)$, where $\mathbf{Q}_{n, k}$ is a symmetric positive semi-definite matrix and depends mainly on the path-loss between BS $n$ and user $k$. 
Throughout the paper, we consider that the users (the receivers) can always estimate their channel vectors with high accuracy, i.e., we consider perfect CSI at the receiver (CSIR). This assumption is justified in practice, as CSIR can be estimated during the training phase with minimal communication overhead \cite{4383372, 10.5555/1111206}. Concerning the CSIT, the CP obtains the channel estimations from the BSs in the network. The BSs acquire the CSI at the beginning of each transmission block, conventionally through uplink training in time division multiplex (TDD) systems \cite{1193803}, or via quantized feedback links in frequency division multiplex (FDD) systems \cite{4641946}. In contrast to CSIR, obtaining high accuracy CSIT requires huge communication overhead and therefore, assuming full CSIT knowledge is somewhat optimistic, and in practice, it might not be possible, especially in dense networks. In this work, we investigate two cases concerning the CSIT assumption:
\begin{itemize} 
	\item{\textbf{{Case 1:}} The CP estimates the channel state perfectly and the error due to quantized feedback or during the uplink training is negligible, i.e., we assume full CSIT. In this case the CP has knowledge of all elements in the vector $\mathbf{h}$}. Obviously, full CSIT case involves a large communication overhead between the users and the CP, which requires a considerable amount of communication resources, which may not be affordable in dense networks.
	\item{\textbf{{Case 2:}} Alternatively, the CP can only access the matrices $\left\lbrace  \mathbf{Q}_{n,k}|\hspace{1mm} n\in \mathcal{N}, k \in \mathcal{K}\right\rbrace $}, i.e., the CP knows the channel statistics of all users. This case is referred to as statistical CSIT, as the CP does not know the channel coefficients $\left\lbrace \mathbf{h}_{n,k}|\hspace{1mm} n \in \mathcal{N}, k \in \mathcal{K} \right\rbrace$ exactly, but their covariance matrix is available to the CP. Note that the perfect estimate of the channel statistics can be easily obtained with minimal communication overhead, as it depends mainly on the user locations, which can be accurately estimated using off-the-shelf global positioning system (GPS) devices \cite{8664604}.      
\end{itemize} 
In the next subsection, we describe the transmission scheme which combines beamforming with RSMA and clustering. 
%  We consider a channel fading model in which the channel state, defined as $\mathbf{h} \triangleq \text{vec}\{\mathbf{h}_k|\,\, \forall k \in \mathcal{K}\}$, varies according to an ergodic stationary process. The channel state is drawn from the support set $\Xi \in \mathbb{C}^{NKL \times 1}$ with the distribution $\mathbb{P}$. We further adopt a block-based transmission model, where each transmission block consists of several time slots. The channel fading coefficients remain constant within one block but may vary independently from one block to another according to a stationary process with the distribution $\mathbb{P}$.  
\subsection{Beamforming, Signal-Construction and Clustering}
The proposed transmission scheme consists of RSMA, BS cluster design for \emph{data-sharing}, and cooperative beamforming to transmit the private and common streams to the intended multicast groups. The data chunk of the file requested by group $g$, i.e., $v_g$, is split at the CP (or locally at the BS $n$) into a private part denoted by $v_{g}^p$ and a common part denoted by $v_{g}^c$.
Afterwards, the CP encodes the private and common parts into $s_{g}^p$ and $s_{g}^c$, respectively, as illustrated in Fig.~\ref{Fig1}. The encoded streams $s_{g}^p$ and $s_{g}^c$ are shared via fronthaul links with BSs that participate in transmitting to the multicast group $g$ and do not cache the requested content. This RSMA strategy is referred to as RSMA and common message decoding (RS-CMD), throughout the paper. Let $\mathbf{w}_{n,g}^p$ and $\mathbf{w}_{n,g}^c$ be the beamforming vectors of BS $n$ to serve the multicast group $g \in \mathcal{G}$. The aggregate beamforming vector from all BSs to a group $g$ can thus be defined as $\mathbf{w}_{g}^o = {\rm vec}\left(\left\{\mathbf{w}_{n,g}^o | \,\, \forall n\in\mathcal{N}\right\}\right)$, where $o\in\{p,c\}$. Note that the beamforming vectors of some BSs not participating in serving group $g$ are vectors where all elements are zero, i.e., $\mathbf{w}_{n,g}^o = \mathbf{0}_L$. Whether a BS $n$ is selected to be serving group $g$ or not is dependent on the proposed clustering algorithm. The respective rates of the private and common streams of multicast group $g \in \mathcal{G}$ are denoted by $R_g^p$ and $R_g^c$. %Thus, we can define $R_g=R_g^p+R_g^c$, where $R_g$ is the achievable rate of multicast group $g$. \\

The BSs that already cached the requested content perform the required baseband processing tasks locally when participating in transmission to the intended multicast groups. However, we emphasize the fact that the beamforming vectors are jointly optimized at the CP to design the coordinated beamforming strategy. After that, the beamforming coefficients, are shared directly with their respective cluster of BSs\footnote{We ignore the fronthaul bandwidth required to support the transmission of beamforming coefficients and we only consider the capacity required to deliver the encoded streams to the BSs in the serving cluster}. \\
Let $\mathcal{G}_n^p,\mathcal{G}_n^c \subseteq \mathcal{G}$ be the subset of groups served by BS $n$ with a private or common message, respectively, i.e.,
\begin{align}
\label{eq:e3}
\mathcal{G}_{n}^{p} & \defeq \left\lbrace g \in \mathcal{G}|\hspace{1mm} \text{BS} \hspace{1mm} n \hspace{1mm} \text{delivers} \hspace{1mm} s_{g}^p \hspace{1mm} \text{to multicast group}\hspace{0.8mm} g\right\rbrace,\\
\label{eq:e4}
\mathcal{G}_{n}^{c} & \defeq \left\lbrace g \in \mathcal{G}|\hspace{1mm} \text{BS} \hspace{1mm} n \hspace{1mm} \text{delivers} \hspace{1mm} s_{g}^c \hspace{1mm} \text{to multicast group}\hspace{0.8mm} g\right\rbrace.
\end{align}
\subsection{Cache Model}
In this paper, we consider storing content closer to the users in local cache memories at the BSs. Associated with this content delivery process, we distinguish between \textit{cache placement} and \textit{cache delivery} phases \cite{8008769,6763007}.
While parts of recent related literature studied efficient designs of \textit{cache placement} strategies to improve the overall content delivery, other works aimed at optimizing the \textit{cache delivery} process for fixed \textit{cache placement}. Generally, we can say that the \textit{cache placement} phase spans over a wider time-scale than the \textit{cache delivery} phase, since the requested content's popularity changes slowly with time. The particular execution of the \textit{cache placement} phase is done in order to significantly improve the \textit{cache delivery} phase, especially during peak-traffic periods.\\
\indent Herein, we aim to optimize the \textit{cache delivery} phase and assume the \textit{cache placement} to be known a priori at the CP \cite{7499119}. Therefore, we let $\mathbf{C} \in \left\lbrace 0,1\right\rbrace^{F \times N} $ be the binary cache placement matrix where $\left[\mathbf{C}\right]_{f, n} = c_{f, n} = 1 $ means that the file $f$ is cached at BS $n$ and $c_{f, n} = 0$ means it is not. All users in the multicast group $g$ request the same file $f_g \in \mathcal{F}$. A cache hit means, a BSs that caches $f_g$, i.e., $c_{f_g, n} = 1$, processes the data locally and there is no need to burden the fronthaul link upon serving group $g$. Otherwise, i.e., $c_{f_g, n} = 0 $, means that in order for $n$ to serve group $g$ it has to receive data from the CP.
Utilizing the binary cache placement matrix has an impact on the mathematical formulation of the fronthaul constraint. Based on the instantaneous rates $R_g^p$ and $R_g^c$ and the serving clusters defined in \eqref{eq:e3} and \eqref{eq:e4}, the fronhaul capacity constraint of all BSs in the proposed C-RAN is
\begin{equation}
	\sum_{{g \subseteq \mathcal{G}_n^p}}(1-c_{f_g, n}){R}_{g}^{p} + \sum_{{g \subseteq \mathcal{G}_n^c}} (1-c_{f_g, n}) {R}_{g}^{c} \leq C_n^{\text{max}}, \qquad \forall n \in \mathcal{N}. \label{eq:front}
\end{equation}
Disregarding the caching ability, all common and private rates of messages served by BS $n$ contribute to the sum and burden the fronthaul link of $n$. This leads to fronthaul congestion, especially when the serving clusters include many groups or the fronthaul capacity is limited. Therefore, in \eqref{eq:front}, if the file requested by group $g$, i.e., $f_g$, is cached by BS $n$, then $c_{f_g,n} = 1$. Thus the common and private rates of $g$ do not contribute to the sum saving fronthaul capacity. Next, we elaborate on RSMA and define our receiver model, also, we further specify the received signal formulation.
\subsection{Receiver Model and Instantaneous Achievable Rates}\label{sec:receivermodel}
In the context of this paper, common messages are employed for the sole purpose of mitigating interference in C-RAN to achieve better resource allocation. Hence, in a C-RAN system that deploys RSMA, each user in a multicast group is expected to decode multiple messages. Thus, the order in which user $k$ decodes the intended messages plays an important role in assessing the efficiency of the relevant proposed interference mitigation techniques. Although joint decoding of all common and private messages at user $k$ would result in optimized rates, it's implementation is complicated in practice. Particularly when the network is large, the intended set of messages to be decoded by each user is large. However, the classical information-theoretical results of a two-user IC already suggest that decoding a strong interferer's common message can significantly improve a user's achievable rate \cite{4675741}. From this perspective, in this paper, we focus on a successive decoding strategy. User $k$ decodes a subset of all common messages in a fixed decoding strategy, based on the descending order of the interferers' channel gains, as described next. 

Each user deploys successive interference cancellation (SIC) to remove parts of the interference in a successive order. A block diagram of the SIC at user $1$ is given in Fig.~\ref{sic}. From Fig.~\ref{sic}, it becomes obvious that the set of common messages that user $k$ is decoding and the order in which the messages are decoded plays an essential role in characterizing the SIC at the users. \\
To this end, let $\mathcal{M}_g$ denote the set of users decoding $s_g^c$, i.e.,
\begin{equation} \label{eq:e2.12}
\mathcal{M}_g \triangleq \left\lbrace j \in \mathcal{K}| \hspace{1mm}\text{user} \hspace{1mm}j\hspace{1mm}  \text{decodes} \hspace{1mm}s_{g}^c \right\rbrace.
\end{equation}%
\begin{figure}
	\centering
	\includegraphics[scale=.7]{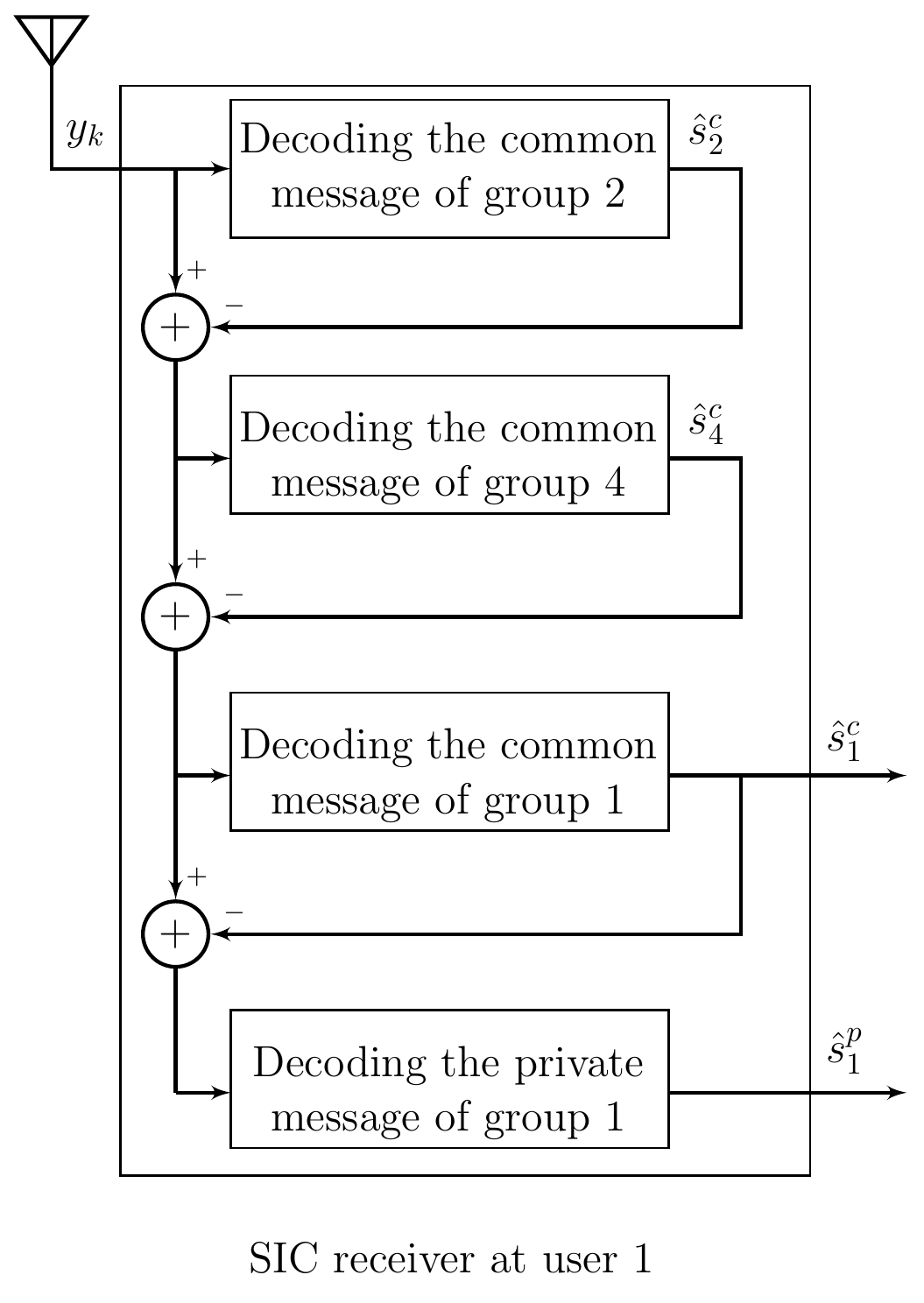}
	\caption{A block diagram for an SIC at user 1. In this example, the common messages decoded at user 1 are $\Phi_{1} = \{1, 2, 4\}$. The decoding order at user 1 is $\pi_1 =  \{2, 4 ,1\}$.}
	\label{sic}
\end{figure}%
%The set of common messages that user $k$ would decode is then given as
Indices of groups whose common messages are decoded by user $k$ are given by
\begin{align} \label{eq:e2.13}
\Phi_{k} \triangleq \left\lbrace g \in \mathcal{G}|\hspace{1mm} k \in \mathcal{M}_{g} \right\rbrace.
\end{align}
%\textcolor{red}{this formulation is incorrect. $\mathcal{G}$ is the set of user group rather than the set of common messages. We need to bring $s_g^c$}\\
The remaining groups whose common messages are not decoded by user $k$ are included in set $\Omega_k$.
We note that once the set $\mathcal{M}_g$ is found, we can determine the set $\Phi_{k}$, and vice-versa.
The choice of $\Phi_{k}$ (and consequently $\mathcal{M}_g$) has a crucial impact on the achievable rate of multicast group $g$.
Consider the following decoding order at user $k$
\begin{equation} \label{eq:e2.14}
\pi_{k}:  \Phi_{k} \rightarrow \left\lbrace 1, 2, \ldots,\left|\Phi_{k}\right|  \right\rbrace,
\end{equation}
which represents a bijective function of the set $\Phi_{k}$ with cardinality $\left|\Phi_{k}\right|$,
i.e., $\pi_{k}(g)$ is the successive decoding step in which the common message of multicast group $g$, i.e., $s_g^c \in \Phi_k$, is decoded at user $k$. {In other terms, $\pi_{k}(g_1) > \pi_{k}(g_2)$ (where $g_1 \neq g_2$) implies that user $k$ decodes the common message of multicast group $g_2$ first, and then the common message of multicast group $g_1$.
	Now, we can rewrite ${y}_k$, the received signal at user $k$ in the multicast group $g$ as follows
	\begin{align}\label{eq:e2.15}
	y_k = & \underbrace{\left(\mathbf{h}_{k}^H\mathbf{w}_{g}^p {s}_{g}^p + \sum_{j \in \Phi_{k}}\mathbf{h}_{k}^H\mathbf{w}_{j}^c {s}_{j}^c\right)}_{\text{Signals to be decoded}} + \underbrace{\sum_{m \in \mathcal{G}\setminus \{g\}}\mathbf{h}_{k}^H\mathbf{w}_{m}^p {s}_{m}^p + \sum_{l \in \Omega_{k}}\mathbf{h}_{k}^H\mathbf{w}_{l}^c {s}_{l}^c+ n_k}_\text{Interference plus noise}.
	\end{align}
	%the set $\Omega_k$ contains the indices of multicast groups whose common messages are not decoded at user $k$.
	User $k$ then uses SIC to remove the common messages in set $\Phi_{k}$ from the received signal $y_k$. The common messages are successively decoded according to the decoding order given by $\pi_{k}$. The common message decoding is solely performed to manage the interference and improve the detectability of the private message, which is decoded last. The average power of the messages received at user $k$ when decoding the private stream $s_g^p$ and the common stream $s_i^c$ of multicast group $i$, respectively, are given as
	\begin{align}
	\label{eq:e2.16}
	T_{g,k}^p & = \left|\mathbf{h}_{k}^{H}\mathbf{w}_{g}^p \right|^2 +  \underbrace{\sum\limits_{j \in \mathcal{G}\setminus \{g\}}\left|\mathbf{h}_{k}^{H}\mathbf{w}_{j}^p \right|^2 + \sum\limits_{l \in \Omega_k}\left|\mathbf{h}_{k}^{H}\mathbf{w}_{l}^c \right|^2  + \sigma^2}_{I_{g,k}^p}, \\
	\label{eq:e2.17}
	T_{i, k}^c & = \left|\mathbf{h}_{k}^{H}\mathbf{w}_{i}^c \right|^2+\underbrace{\sum_{j \in \mathcal{G}}\left|\mathbf{h}_{k}^{H}\mathbf{w}_{j}^p \right|^2 + \sum\limits_{l \in \Omega_k}\left|\mathbf{h}_{k}^{H}\mathbf{w}_{l}^c \right|^2 + \sum\limits_{m  \in \Psi_{i, k}}\left|\mathbf{h}_{k}^{H}\mathbf{w}_{m}^c \right|^2 +\sigma^2}_{I_{i, k}^c},
	\end{align}
	where $\Psi_{i, k} \triangleq \left\lbrace m \in \Phi_k | \pi_{k}(m)> \pi_{k}(i) \right\rbrace$. $\Psi_{i, k}$ is a set including indices of all multicast groups $m$ whose common messages are decoded after decoding the common message of group $i$, i.e., groups having common messages with higher decoding order $\pi_{k}(m)>\pi_{k}(i)$. These multicast groups $m$ contribute to the interference since they have yet to be decoded. Also, $I_{g,k}^p$ and $I_{i,k}^c$ denote the interference-plus-noise at user $k$ decoding the private message of multicast group $g$ and the common message of group $i$, respectively.	
	Based on the expressions in \eqref{eq:e2.16} and \eqref{eq:e2.17}, we define the SINR of user $k$ from multicast group $g$, i.e., $k\in\mathcal{G}_g$, when decoding the private message of group $g$ and the SINR of user $k$ when decoding the common message of group $i$ as\\
	\begin{align}
	\label{eq:e2.18}
	\gamma_{g,k}^p &= \frac{\left|\mathbf{h}_{k}^{H}\mathbf{w}_{g}^p \right|^2}{\sigma^2 + \sum\limits_{j \in \mathcal{G}\setminus \{g\}}\left|\mathbf{h}_{k}^{H}\mathbf{w}_{j}^p \right|^2 + \sum\limits_{l \in \Omega_g}\left|\mathbf{h}_{k}^{H}\mathbf{w}_{l}^c \right|^2},\\
	\label{eq:e2.19}
	\gamma_{i, k}^c &= \frac{\left|\mathbf{h}_{k}^{H}\mathbf{w}_{i}^c \right|^2}{\sigma^2+\sum\limits_{j \in \mathcal{G}}\left|\mathbf{h}_{k}^{H}\mathbf{w}_{j}^p \right|^2 + \sum\limits_{l \in \Omega_{i}}\left|\mathbf{h}_{k}^{H}\mathbf{w}_{l}^c \right|^2 + \sum\limits_{m  \in \Psi_{i, k}}\left|\mathbf{h}_{k}^{H}\mathbf{w}_{m}^c \right|^2} .%\CommaPunct
	\end{align}   
	The instantaneous achievable rate of multicast group $g$ is given as $R_g = R_{g}^p + R_{g}^c$, where the instantaneous private and common rates satisfy the following achievability conditions
	\begin{align}
	\gamma_{g,k}^p &\geq 2^{R_{g}^p/B} - 1, \quad \forall k \in \mathcal{G}_g, \forall g \in \mathcal{G}, \label{eq:e2.20}\\
	\gamma_{g,k}^c &\geq 2^{R_{g}^c/B} - 1, \quad \forall k \in \mathcal{M}_g, \forall g \in \mathcal{G}. \label{eq:e2.21}
	\end{align}
	Note that the interference from sending the common message $s_g^c$ does not impact the users in $\mathcal{M}_g$ as they decode this message. This is the main motivation for employing RSMA in networks that suffer from interference. We emphasize that the instantaneous rate constraints are achievable under the assumption of full CSIT. However, when the CP knows only the channel's statistical properties, i.e., statistical CSIT, the achievability constraints \eqref{eq:e2.20} and \eqref{eq:e2.21} are not valid, as the SINR expressions become functions of random variables. In this case, we assort instead to the ergodic achievable rate for sending private and common messages, as discussed next. 
	\subsection{Achievable Ergodic Rates}
	\label{subsec:rec}
	In the statistical CSIT scenario, we consider the CP has only information about the statistics of channel states. %\textcolor{red}{However, the receivers are considered to have perfect knowledge of the channel states, i.e., we assume perfect CSIR is available at users so that user $k$ knows the channel state $\mathbf{h}_k$ perfectly.} 
	These assumptions are quite general and can model other inaccuracies in CSIT. 
	%concerning the statistical \gls{csit} model, the channel between \glssc{bs} $n$ and user $k$ follows the distribution $\mathbf{h}_{n, k} \sim \mathbf{\mathcal{CN}}\left(0, \mathbf{Q}_{n, k}\right)$. $\mathbf{Q}_{n, k}$ is a symmetric positive semi-definite matrix and represents the channel covariance matrix. 
	The assumption of statistical CSI knowledge, in particular, is reasonable because the path-loss information varies slowly and needs to be updated when the users' location changes only, which significantly reduces the communication overhead due to the CSIT acquisition process at the CP compared to the case in which the CP acquires full CSIT.
	%In what follows, by imperfect \gls{csit}, we mean that the \gls{cp} is only aware of the covariance matrices, while full \gls{csit} implies that the \gls{cp} knows the coefficients $\left\lbrace \mathbf{h}_{n,k}| \forall k \in \mathcal{K}, \forall n \in \mathcal{N}\right\rbrace $ perfectly
	%\textcolor{red}{With full CSIT at the CP, we can adapt the beamforming vectors and, eventually, the transmit rate to each channel state. Obviously, with full CSIT, we can achieve the best possible rate to send the streams to users as the network resources are adapted to each channel state, i.e., the beamforming vectors, the serving cluster of BSs and the allocated rate per-stream.} \textcolor{red}{However, with a lack of CSIT, the transmitter can not adapt the beamforming vectors and the rates to each channel state as it is not known at the transmitter.} 
	In this case we consider sending the private and common streams of group $g$ at the ergodic rate \cite{goldsmith_2005}. The total ergodic rate of group $g$ is defined as $\EX_{\mathbf{h}}\left\lbrace R_{g}^{p} + R_{g}^{c} \right\rbrace \triangleq \bar{R}_{g}^{p} + \bar{R}_{g}^{c}$, where $\bar{R}_{g}^{p}$ is the ergodic rate to send the private stream and $\bar{R}_{g}^{c}$ is the ergodic rate to send the common stream of multicast group $g$. The achievability relations of the ergodic private and common rates are 
	\begin{align}
	\label{eq:e2.22}
	\bar{R}_{g}^{p} \leq B\EX_{\mathbf{h}}\left\lbrace  \log_2\left(1 + \gamma_{g,k}^p\right)\right\rbrace,  &\quad  \forall k \in \mathcal{G}_g, \forall g \in \mathcal{G}, \\
	\label{eq:e2.23}
	\bar{R}_{g}^{c} \leq B\EX_{\mathbf{h}}\left\lbrace \log_2\left(1 + \gamma_{g, k}^c\right)\right\rbrace,  &\quad \forall k \in \mathcal{M}_g, \forall g \in \mathcal{G}.
	\end{align}
	In the next section, we formulate the optimization problem under consideration, propose our solution approach. To that end, we define the WMMSE-based algorithm for optimal resource allocation in the considered system model.
\section{Problem Formulation and Proposed Algorithms}\label{sec:PF}
This paper considers the problem of system utility maximization in a cache-assisted C-RAN which adopts a multicast group-based transmission scheme combining RSMA and \emph{data-sharing}. The constraints of the problem consist of per-BS limited fronthaul capacity and transmit power constraints, as well as per-stream achievable rate constraints. The fronthaul link is needed to fetch data from the CP for multicast groups whose content is not locally cached at the BSs. The system utility under focus is the minimum achievable rate amongst all multicast groups. Thus, the resources in C-RAN are allocated to guarantee fairness among all multicast groups being served in C-RAN. The optimization problem to realize this target can be defined as follows
\begin{subequations}\label{eq:Opt0}
	\begin{align}
		\underset{\mathcal{V}_0}{\text{max}}\quad &\bar{R}  \\
		\text{s.t.} \quad & \bar{R} \leq \bar{R}_g^p + \bar{R}_g^c,  && \forall g \in \mathcal{G}, \label{eq:R1} \\
		&\sum_{{g \in \mathcal{G}_n^p}}\norm[\big]{\mathbf{w}_{n ,g}^{p}}_2^2  + \sum_{{g \in \mathcal{G}_n^c}}\norm[\big]{\mathbf{w}_{n ,g}^{c}}_2^2\leq P_{n}^{\text{max}},  &&\forall n \in \mathcal{N}, \label{eq:pmax1}\\
		&\sum_{{g \subseteq \mathcal{G}_n^p}}(1-c_{f_g, n})\bar{R}_{g}^{p} + \sum_{{g \subseteq \mathcal{G}_n^c}} (1-c_{f_g, n}) \bar{R}_{g}^{c} \leq C_n^{\text{max}}, &&\forall n \in \mathcal{N}, \label{eq:fn1}\\
		& \bar{R}_{g}^{p} \leq \mathbb{E}_{\mathbf{h}}\left\lbrace B\log_2(1+\gamma_{g,k}^p)\right\rbrace, &&\forall k \in \mathcal{G}_g, \forall g \in \mathcal{G}, \label{eq:Rp1} \\	
		& \bar{R}_{g}^{c} \leq \mathbb{E}_{\mathbf{h}}\left\lbrace B\log_2(1+\gamma_{g,k}^c)\right\rbrace, &&\forall k \in \mathcal{M}_g, \forall g \in \mathcal{G}, \label{eq:Rc1}
	\end{align}
\end{subequations}  
where $\mathcal{V}_0 \triangleq \left\lbrace \mathbf{w}_{g}^p, \mathbf{w}_{g}^c, \bar{R}_{g}^{p}, \bar{R}_{g}^{c}, \mathcal{G}_n^p, \mathcal{G}_n^c|\,\, \forall g \in \mathcal{G}, \forall n \in \mathcal{N}\right\rbrace$ is the set of optimization variables. Here we maximize the minimum achievable rate of all multicast groups $\bar{R}$ w.r.t. the private and common beamforming vectors, ergodic rates, and serving clusters. Constraint \eqref{eq:R1} represents the quality of service (QoS) ergodic rate minimum constraint, where $\bar{R}_g^p + \bar{R}_g^c$ is the ergodic achievable rate of multicast group $g$ and $\bar{R}$ is the minimum ergodic rate. \eqref{eq:pmax1} is the maximum transmit power constraint per BS, where $\mathbf{w}_{n,g}^p$ and $\mathbf{w}_{n,g}^c$ are the beamforming vectors from BS $n$ to serve the private and common message of multicast group $g$, respectively. $\mathcal{G}_n^p$ and $\mathcal{G}_n^c$ are the subsets of groups whose private and common messages are served by BS $n$, respectively. \eqref{eq:fn1} is the available fronthaul capacity constraint per BS. The achievable ergodic rates of the private and common streams are given by constraints \eqref{eq:Rp1} and \eqref{eq:Rc1}, respectively. \\
%\begin{subequations}\label{eq:Opt0}
%	\begin{alignat}{5} %\hspace{0.5mm}\underset{g \in \mathcal{G}}{\text{min}}
%	{P}_{0}:\quad&\underset{\mathcal{V}_0}{\text{max}} && \bar{R} \nonumber \\
%	&\text{subject to} \quad&& \eqref{eq:e2.18} \, \text{and} \, \eqref{eq:e2.19}, \nonumber\\ %\quad \eqref{ToRate}\\
%	&&& \bar{R} \leq \bar{R}_{g}^{p} + \bar{R}_{g}^{c} && g \subseteq \mathcal{G}, \label{eq:e18a} \\
%	&&&\sum_{{g \in \mathcal{G}_n^p}}\normV[\big]{\mathbf{w}_{n ,g}^{p}}_2^2  + \sum_{{g \in \mathcal{G}_n^c}}\normV[\big]{\mathbf{w}_{n ,g}^{c}}_2^2\leq P_{n}^{\text{max}}, \quad &&\forall n \in \mathcal{N}, \label{eq:e18b}\\
%	&&& \sum_{{g \subseteq \mathcal{G}_n^p}}(1-c_{f_g, n})\bar{R}_{g}^{p} + \sum_{{g \subseteq \mathcal{G}_n^c}} (1-c_{f_g, n}) \bar{R}_{g}^{c} \leq F_n, \quad &&\forall n \in \mathcal{N}, \label{eq:e18c}\\
%	&&& \bar{R}_{g}^{p} \leq \mathbb{E}_{\mathbf{h}}\left\lbrace B\log_2(1+\gamma_{g,k}^p)\right\rbrace, \quad &&\forall k \in g, g \subseteq \mathcal{G}, \label{eq:e18d} \\
%	&&& \bar{R}_{g}^{c} \leq \mathbb{E}_{\mathbf{h}}\left\lbrace B\log_2(1+\gamma_{g,k}^c)\right\rbrace,  \quad && \forall k \in \mathcal{M}_g, \forall g \subseteq \mathcal{G}, \label{eq:e18e}
%	\end{alignat}
%\end{subequations}
%where $\mathcal{V}_0$ is the set of optimization variables is given as
%\begin{equation}\label{eq:e3.44}
%\mathcal{V}_0 \triangleq \left\lbrace \mathbf{w}_{g}^p, \mathbf{w}_{g}^c, \bar{R}_{g}^{p}, \bar{R}_{g}^{c}, \mathcal{G}_n^p, \mathcal{G}_n^c|\,\, \forall g \in \mathcal{G}, \forall n \in \mathcal{N}\right\rbrace.
%\end{equation}
\indent With stochastic coordinated beamforming optimization, the same beamforming vectors, i.e., $\left\lbrace \mathbf{w}_{g}^p, \mathbf{w}_{g}^c|\,\, \forall g \in \mathcal{G}\right\rbrace$, are used for all transmit blocks in which the channel statistics remain constant. The same applies to the serving clusters and the allocated ergodic rates, which remain unchanged over several transmission blocks in which the channel statistics do not alter. Problem \eqref{eq:Opt0} is difficult and challenging to solve. In particular, the constraints \eqref{eq:Rp1} and \eqref{eq:Rc1} are functions of a stochastic quantity.
Hence, the achievable rates depend on the current realization of channel fading which is unknown at the CP and the expectations in \eqref{eq:Rp1} and \eqref{eq:Rc1} have no closed-form. Moreover, even when considering the deterministic version of this problem, i.e., when assuming full CSIT, the SINR expressions are non-convex functions of the design variables (i.e., the beamforming vectors) and the resulting problem is known to be NP-hard. Therefore, to tackle problem \eqref{eq:Opt0}, we require other optimization tools than those used for solving the counterpart deterministic problem. We use a three step approach: First, we develop a group-based clustering algorithm to predetermine the serving clusters, i.e., $\mathcal{G}_n^p$ and $\mathcal{G}_n^c$, in order to simplify problem \eqref{eq:Opt0}. Then, we use the SAA method to reformulate the expected value expressions. Afterwards, we make use of a WMMSE-rate relationship, to construct a block coordinate ascent algorithm for solving the resulting continuous, deterministic NLP.     
\subsection{Group-based Clustering}
The authors in \cite{DaiY14} propose a clustering algorithm for assigning a set of BSs to serve each user in the network. It is based on the path loss, which depends on the user location and typically varies on a slow time scale. However, the algorithm is not directly applicable to our problem since it was developed for unicast streams. Specifically, the clustering algorithm originally proposed in \cite{DaiY14} is a user-centric clustering approach not suited for the group-based transmission in this work. The authors in \cite{9217249} extended this algorithm to a more general case, as they considered common messages to be decoded by a subset of users. However, the private messages in \cite{9217249} are also decoded at a single user only. To further extend the algorithm from \cite{9217249}, in this subsection, we propose a group-based, rather than a user-based, clustering algorithm suited for the proposed multicast transmission scheme, which is more general since groups may consist of multiple users. %In this subsection, we will propose an adapted algorithm for group-based clustering under the proposed multicast-stream scheme. 

First, we define $A_n^{\text{max}}$ as the maximum number of streams a BS $n$ can serve. Such a constant helps balancing the load, as it prevents BSs from being overloaded. Our group-based clustering approach starts by finding candidate clusters of BSs to serve each multicast group. These clusters are denoted as $\mathcal{N}_g^p$ and $\mathcal{N}_g^c$, containing all BSs that serve the private and common stream of multicast group $g$, respectively. We emphasize the multicast nature of our algorithm, as every BS in a serving cluster $\mathcal{N}_g^p$ or $\mathcal{N}_g^c$ has to have a good channel quality to all users in the group. Since the algorithm is based on path loss, we use a channel quality measure $q_{n,k}$ from BS $n$ to user $k$ that is inversely proportional to the path loss. As aforementioned, we can not use this measure directly in our algorithm. Hence, we define the collective channel quality of all users decoding the private and common message of group $g$ as $\tilde{q}_{n,g}^p$ and $\tilde{q}_{n,g}^c$, respectively. This collective quality is computed as 
\begin{equation}
	\tilde{q}_{n,g}^p = \frac{1}{|\mathcal{G}_g|}\sum_{j\in\mathcal{G}_g} q_{n,j}, \qquad \tilde{q}_{n,g}^c = \frac{1}{|\mathcal{M}_g|}\sum_{j\in\mathcal{M}_g} q_{n,j}.
\end{equation}
Thus, the candidate clusters of BSs serving a group $g$ are given by
\begin{equation}
	\mathcal{N}_g^p = \left\{ n \, \Big|\, \underset{m\in\mathcal{N}}{\text{max}}( \tilde{q}_{m,g}^p ) - \tilde{q}_{n,g}^p \leq \mu \right\}, \quad \mathcal{N}_g^c = \left\{ n \, \Big|\, \underset{m\in\mathcal{N}}{\text{max}}( \tilde{q}_{m,g}^c ) - \tilde{q}_{n,g}^c \leq \mu \right\}. \label{eq:ngo}
\end{equation}
%Hence, we define the collective channel quality of all users decoding the message of group $g$ as $\tilde{q}_{n,\mathcal{D}_g^o}$, where $o\in\{p,c\}$ and \textcolor{red}{$\mathcal{D}_g^o$ being the users decoding the private/common message of group $g$}, i.e., $\mathcal{D}_g^p = \mathcal{G}_g$ and  $\mathcal{D}_g^c = \mathcal{M}_g$. This collective quality is computed as 
%\begin{equation}
%	\tilde{q}_{n,\mathcal{D}_g^o} = \frac{1}{|\mathcal{D}_g^o|}\sum_{j\in\mathcal{D}_g^o} q_{n,j}, \qquad o\in\{p,c\}.
%\end{equation}
%Thus, the candidate clusters of BSs serving a group $g$ are given by \textcolor{red}{check}
%\begin{equation}
%	\mathcal{N}_g^o = \left\{ n \, \Big|\, \underset{m\in\mathcal{N}}{\text{max}}( \tilde{q}_{m,\mathcal{D}_g^o} ) - \tilde{q}_{n,\mathcal{D}_g^o} \leq \mu \right\},\quad o\in\{p,c\}. \label{eq:ngo}
%\end{equation}
The detailed steps of the group-based clustering procedure are listed in Algorithm \ref{alg_clust}. First step is determining the candidate clusters of BSs for all multicast groups according to \eqref{eq:ngo}. This is done for each group's private and common message, respectively. Further, sets $\mathcal{S}$ and $\mathcal{N}$ are initialized, where the former contains all private and common streams and the latter contains all BSs. The algorithm outputs the serving clusters $\mathcal{G}_n^p$ and $\mathcal{G}_n^c$, initialized as empty sets. The first \emph{for}-loop assigns BS $n$ from $\mathcal{N}_g^o$ with the best collective channel quality to serve the stream $s_g^o$. Thus, $g$ is added to the serving cluster of BS $n$, while $n$ is removed from the candidate cluster of stream $s_g^o$. If there is no candidate to serve $s_g^o$, the message is removed from $\mathcal{S}$. In the second \emph{for}-loop, for each BS $n$, we check if $n$ is overloaded, i.e., the number of assigned messages in $n$'s serving cluster exceeds $A_n^{\text{max}}$. If it is overloaded, $|\mathcal{G}_n^p|+|\mathcal{G}_n^c|-A_n^{\text{max}}$ groups with the weakest collective channel quality are removed from the serving cluster. BS $n$ is then removed from $\mathcal{N}$ and from all candidate clusters. This procedure is repeated until all BSs reach their maximum capacity of messages, i.e., $A_n^{\text{max}}$, or all messages $s_g^o$ run out of BS candidates. \\
\begin{algorithm}[h]
	\caption{Group-based Clustering Algorithm}
	\begin{spacing}{1}
		\begin{algorithmic}[1]
			\STATE Set $A_n^{\text{max}}$, determine $\mathcal{N}_g^p$ and $\mathcal{N}_g^c$ by \eqref{eq:ngo} for all multicast groups
			\STATE Initialize $\mathcal{S}=\{s_g^p,s_g^c|\forall g\in\mathcal{G}\}$, $\mathcal{N}=\{1,2,\cdots,N\}$, $\mathcal{G}_n^p = \emptyset$, $\mathcal{G}_n^c = \emptyset$ $\forall n\in\mathcal{N}$
			\STATE \textbf{While} $\mathcal{S}\neq\emptyset \cup \mathcal{N}\neq\emptyset$ \textbf{do} 
			\STATE $\;\;$\textbf{For} $g\in\mathcal{G}, o\in\{p,c\}$ 
			\STATE $\;\;$$\;\;$\textbf{If} $\mathcal{N}_g^o\neq\emptyset$ 
			\STATE $\;\;$$\;\;$$\;\;$The strongest BS from $\mathcal{N}_g^o$ is assigned to serve $s_g^o$, i.e., $\mathcal{G}_n^o = \mathcal{G}_n^o \cup \{g\}$, $\mathcal{N}_g^o = \mathcal{N}_g^o\backslash\{n\}$
			\STATE $\;\;$$\;\;$\textbf{Else} 
			\STATE $\;\;$$\;\;$$\;\;$$\mathcal{S}=\mathcal{S}\backslash\{s_g^o\}$
			\STATE $\;\;$$\;\;$\textbf{End} 
			\STATE $\;\;$\textbf{End} 
			\STATE $\;\;$\textbf{For} $n\in\mathcal{N}$ 
			\STATE $\;\;$$\;\;$\textbf{If} $|\mathcal{G}_n^p|+|\mathcal{G}_n^c|>A_n^{\text{max}}$, i.e., the BS $n$ is overloaded
			\STATE $\;\;$$\;\;$$\;\;$Remove $x$ weakest groups, where $x = |\mathcal{G}_n^p|+|\mathcal{G}_n^c|-A_n^{\text{max}}$
			\STATE $\;\;$$\;\;$$\;\;$$\mathcal{N}=\mathcal{N}\backslash\{n\}$, $\mathcal{N}_g^p=\mathcal{N}_g^p\backslash\{n\}$, $\mathcal{N}_g^c=\mathcal{N}_g^c\backslash\{n\}$ $\forall g\in\mathcal{G}$
			\STATE $\;\;$$\;\;$\textbf{End} 
			\STATE $\;\;$\textbf{End} 
			\STATE \textbf{End while}  
			%\EndWhile
		\end{algorithmic}
	\end{spacing}
	\label{alg_clust}
\end{algorithm}%
\indent At this point, we utilize Algorithm \ref{alg_clust} to pre-compute the serving clusters. Therefore, we can now fix $\mathcal{G}_n^p$ and $\mathcal{G}_n^c$, which makes problem \eqref{eq:Opt0} more traceable from an optimization perspective. The relaxed problem with serving clusters already pre-computed can be stated as
\begin{subequations}\label{eq:Opt01}
	\begin{align}
		\underset{\mathcal{V}_1}{\text{max}}\quad &\bar{R}  \\
		\text{s.t.} \quad & \eqref{eq:R1}, \eqref{eq:pmax1}, \eqref{eq:fn1}, \eqref{eq:Rp1}, \eqref{eq:Rc1}.
	\end{align}
\end{subequations}  
Now, $\mathcal{V}_1 \triangleq \left\lbrace \mathbf{w}_{g}^p, \mathbf{w}_{g}^c, \bar{R}_{g}^{p}, \bar{R}_{g}^{c}|\,\, \forall g \in \mathcal{G}\right\rbrace$ is the new set of optimization variables. Next, we describe the SAA of problem \eqref{eq:Opt01}.
\subsection{SAA Reformulation}
A similar optimization framework which combines SAA and WMMSE-based optimization is first proposed in \cite{7555358}, where it was developed to solve a sum-rate maximization problem for the MISO-BC. %   The following steps for reformulating the stochastic problem \eqref{eq:Opt01} into a deterministic problem using SAA and then reformulating the resulting problem using a WMMSE rate relationship to formulate a WMMSE-based optimization algorithm are first proposed in \cite{7555358}.
However, in this work, we generalize the algorithm to solve a problem of maximizing the minimum rate for multigroup multicast transmission in cache-assisted C-RAN with RSMA. First, we approximate the expected value in the achievable ergodic rate constraints \eqref{eq:Rp1} and \eqref{eq:Rc1} with the sample average \cite{saa}. We use a Monte-Carlo sample size of $M$. With $\mathbf{h}$ as a random vector, we use $\mathbf{h}^m$ to denote the $m$-th independent realization of $\mathbf{h}$. Consequently, we denote the SINR's from \eqref{eq:e2.20} and \eqref{eq:e2.21} as $\gamma_{g,k}^p(m)$ and $\gamma_{g,k}^c(m)$, respectively. The corresponding SAA reformulation of problem \eqref{eq:Opt01} is defined as follows
\begin{subequations}\label{eq:Opt1}
	\begin{align}
		\underset{\mathcal{V}_1}{\text{max}}\quad &\bar{R}  \\
		\text{s.t.} \quad & \eqref{eq:R1}, \eqref{eq:pmax1}, \eqref{eq:fn1}, \nonumber \\
		& \bar{R}_{g}^{p} - \frac{B}{M} \sum_{m=1}^{M} \log_2(1+\gamma_{g,k}^p(m)) \leq 0, &&\forall k \in \mathcal{G}_g, \forall g \in \mathcal{G}, \label{eq:Rp2}\\	
		& \bar{R}_{g}^{c} - \frac{B}{M} \sum_{m=1}^{M} \log_2(1+\gamma_{g,k}^c(m)) \leq 0, &&\forall k \in \mathcal{M}_g, \forall g \in \mathcal{G}. \label{eq:Rc2}
	\end{align}
\end{subequations} 
At this point, problem \eqref{eq:Opt1} is still non-convex. An additional burden is the dependency on the sample size $M$. Nevertheless, problem \eqref{eq:Opt1} now comprises the deterministic constraints \eqref{eq:Rp2} and \eqref{eq:Rc2}, which denote the SAA of \eqref{eq:Rp1} and \eqref{eq:Rc1}, respectively. Note that in the asymptotic regime, the global optimal solutions of problem \eqref{eq:Opt01} and problem \eqref{eq:Opt1} converge. %\textcolor{blue}{Also, note that a special case of this problem reformulation and the resulting algorithm was proposed in \cite{7555358}. However, in this work, we generalize the reformulation and the algorithm to cache-aided C-RANs with multigroup multicasting transmission under statistical CSIT.}
\begin{theorem}\label{th1}
	For $M\rightarrow\infty$ and in the asymptotic regime, the global optimal solution of problem \eqref{eq:Opt1} converges to the global optimal solution of problem \eqref{eq:Opt01}, which is of stochastic nature.
\end{theorem}
\begin{proof}
	Please refer to Appendix \ref{app1}.
\end{proof}
Problem \eqref{eq:Opt1} is still non-convex and the current formulation works for full CSIT only. In more details, constraints \eqref{eq:Rp2} and \eqref{eq:Rc2} are dependent on the explicit definition of the SINR variables $\gamma_{g, k}^p(m)$ and $\gamma_{i, k}^c(m)$, respectively. Since these definitions include the channel vector $\mathbf{h}_k$ for all users, full CSIT is required here. To cope with the statistical CSIT scenario, problem \eqref{eq:Opt1} has to be further reformulated. Therefore, in the next subsection, we will investigate the mean square error (MSE) and find receiver coefficients to minimize the MSE. To that end, we establish a WMMSE-rate relationship to alleviate the current achievable rate constraints.
\subsection{WMMSE Rate Relationship}
The MSE for user $k$ decoding the private message of it's respective multicast group $g$ can be defined as
\begin{equation}\label{eq:egkp0}
	e_{g,k}^p = \mathds{E}\left\{ \left| u_{g,k}^p \left(y_k - \sum_{j\in\Phi_k} \mathbf{h}_k^H \mathbf{w}_j^c s_j^c\right) - s_g^p \right|^2 \right\}.
\end{equation}
Here, $u_{g,k}^p$ is the linear receiver coefficient at $k$ decoding $g$'s private message. The middle term in the parenthesis represents the received signal after canceling all common messages $j$ decoded by user $k$, i.e., $j\in\Phi_k$. In a similar manner, we define the MSE for user $k$ decoding the common message of multicast group $i$ as
\begin{equation}\label{eq:egkc0}
	e_{i,k}^c = \mathds{E}\left\{ \left| u_{i,k}^c \left(y_k - \sum_{m\in\tilde{\Psi}_{i,k}} \mathbf{h}_k^H \mathbf{w}_m^c s_m^c\right) - s_i^c \right|^2 \right\},
\end{equation}
where $\tilde{\Psi}_{i,k} = \{ m\in\Phi_k | \pi_k(i) > \pi_k(m) \}$. Please note the difference between $\tilde{\Psi}_{i,k}$ and $\Psi_{i,k}$. The former $\tilde{\Psi}_{i,k}$ contains all multicast groups (indexed by $m$) whose common messages are decoded before decoding group $i$'s common message, i.e., groups having messages with lesser decoding order $\pi_{k}(i)>\pi_{k}(m)$. That is, the common messages of these groups are not included in the received signal since the are already decoded. In contrast, $\Psi_{i,k}$ includes groups whose common messages are yet to be decoded. Also, $u_{i,k}^c$ is $k$'s receiver coefficient decoding the common message of multicast group $i$. Utilizing the definition of $y_k$ from \eqref{eq:e2.15}, we can transform \eqref{eq:egkp0} and \eqref{eq:egkc0} into
\begin{align}
	e_{g,k}^p &= |u_{g,k}^p|^2 T_{g,k}^p - 2 \text{Re}\{ u_{g,k}^p \mathbf{h}_k^H \mathbf{w}_g^p \} + 1, \label{eq:egkp}\\
	e_{i,k}^c &= |u_{i,k}^c|^2 T_{i,k}^c - 2 \text{Re}\{ u_{i,k}^c \mathbf{h}_k^H \mathbf{w}_i^c \} + 1. \label{eq:egkc}
\end{align}
Receiver coefficients that minimize the MSE can be found computing $\frac{\partial e_{g,k}^p}{\partial u_{g,k}^p} = 0$ and $\frac{\partial e_{i,k}^c}{\partial u_{i,k}^c} = 0$. To this end, we obtain the minimum MSE (MMSE) receiver coefficients
\begin{align}
	u_{g,k,\text{mmse}}^p &= (\mathbf{w}_g^p)^H \mathbf{h}_k / T_{g,k}^p, \label{eq:upmmse}\\
	u_{i,k,\text{mmse}}^c &= (\mathbf{w}_i^c)^H \mathbf{h}_k / T_{i,k}^c. \label{eq:ucmmse}
\end{align}
Hence, to obtain the MMSE terms, \eqref{eq:upmmse} and \eqref{eq:ucmmse} are inserted into \eqref{eq:egkp} and \eqref{eq:egkc}, respectively. The MMSE at user $k$ decoding the private message of group $g$ is $e_{g,k,\text{mmse}}^p = I_{g,k}^p / T_{g,k}^p$. Equivalently, $e_{i,k,\text{mmse}}^c = I_{i,k}^c / T_{i,k}^c$ denotes the MMSE at $k$ decoding $i$'s common message. Next, we observe a specific relation between the achievable rate and the MMSE in a form that is amenable for the WMMSE-based algorithm.
\begin{lemma}
	The achievable rates from \eqref{eq:e2.20} and \eqref{eq:e2.21} can be expressed in another from including the receiver coefficients, the error terms, and error weight variables as follows
	\begin{align}
		\log_2(1+\gamma_{g,k}^p) &= \underset{u_{g,k}^p,\rho_{g,k}^p}{\text{max}} \left( \frac{\log(\rho_{g,k}^p)-\rho_{g,k}^p e_{g,k}^p + 1}{\log(2)} \right), \quad \forall k \in \mathcal{G}_g, \forall g \in \mathcal{G} \label{eq:mseratep},\\
		\log_2(1+\gamma_{g,k}^c) &= \underset{u_{g,k}^c,\rho_{g,k}^c}{\text{max}} \left( \frac{\log(\rho_{g,k}^c)-\rho_{g,k}^c e_{g,k}^c + 1}{\log(2)} \right), \quad \forall k \in \mathcal{M}_g, \forall g \in \mathcal{G} \label{eq:mseratec}.
	\end{align}
	Hereby we introduce $\rho_{g,k}^p$ and $\rho_{g,k}^c$ as weights for the MSE terms.
\end{lemma}
\begin{proof}
	The partial derivative of the right hand side of \eqref{eq:mseratep} w.r.t. $u_{g,k}^p$ is set to zero. We obtain the optimal receiver coefficient at $k$ decoding $g$'s common message as $(u_{g,k}^p)^* = u_{g,k,\text{mmse}}^p$, which is the optimal MSE receiver coefficient from \eqref{eq:upmmse}. The same procedure for the weight coefficient $\rho_{g,k}^p$ results in $(\rho_{g,k}^p)^* = 1/{e_{g,k,\text{mmse}}^p}$. Inserting the optimal coefficients into \eqref{eq:mseratep} results in the equivalence of the right hand and left hand side. To prove the equivalence in \eqref{eq:mseratec}, the same procedure can be repeated. This completes the proof.
%	\begin{align}
%		(u_{g,k}^p)^* &= u_{g,k,\text{mmse}}^p, \quad (\rho_{g,k}^p)^* = \frac{1}{e_{g,k,\text{mmse}}^p} \nonumber\\
%		(u_{i,k}^c)^* &= u_{i,k,\text{mmse}}^c, \quad (\rho_{i,k}^c)^* = \frac{1}{e_{i,k,\text{mmse}}^c} \nonumber
%	\end{align}
\end{proof}
At this point, we obtain a formulation for the achievable rate, which helps in formulating the WMMSE-based algorithm. Note that for full CSIT, we could now have started discussing the final algorithm. However, as we consider statistical CSIT, the achievable rate relations in \eqref{eq:e2.20} and \eqref{eq:e2.21} become non-deterministic functions. Therefore, we next elaborate on the expressions \eqref{eq:mseratep} and \eqref{eq:mseratec} in the context of statistical CSIT. \\
\indent Due to the lack of of channel coefficient knowledge, we can not use the rate expressions \eqref{eq:e2.20} and \eqref{eq:e2.21}. Hence, we define the achievability relations of the ergodic private and common rates in \eqref{eq:e2.22} and \eqref{eq:e2.23}, respectively. In a similar manner, we reformulate \eqref{eq:mseratep} and \eqref{eq:mseratec} as follows
\begin{align}
	\mathds{E}_\mathbf{h}\{\log_2(1+\gamma_{g,k}^p)\} &= \frac{1}{\log(2)}\mathds{E}_\mathbf{h}\bigg\{\underset{u_{g,k}^p,\rho_{g,k}^p}{\text{max}} \left( \log(\rho_{g,k}^p)-\rho_{g,k}^p e_{g,k}^p + 1 \right)\bigg\}, \; \forall k \in \mathcal{G}_g, \forall g \in \mathcal{G}, \label{eq:wmmseratep}\\
	\mathds{E}_\mathbf{h}\{\log_2(1+\gamma_{g,k}^c)\} &= \frac{1}{\log(2)}\mathds{E}_\mathbf{h}\bigg\{\underset{u_{g,k}^c,\rho_{g,k}^c}{\text{max}} \left( \log(\rho_{g,k}^c)-\rho_{g,k}^c e_{g,k}^c + 1 \right) \bigg\} , \; \forall k \in \mathcal{M}_g, \forall g \in \mathcal{G} \label{eq:wmmseratec}.
\end{align}
In the next subsection, we will describe the following reformulation steps and, to that end, state the WMMSE-based algorithm.
\subsection{WMMSE-Based Algorithm}
To substitute the WMMSE rate relationship into the constraints \eqref{eq:Rp2} and \eqref{eq:Rc2}, we apply SAA to the formulations from \eqref{eq:wmmseratep} and \eqref{eq:wmmseratec}, respectively. Thus, we can define the reformulated optimization problem as
\begin{subequations}\label{eq:Opt2}
	\begin{align}
		\underset{\mathcal{V}_2}{\text{max}}\quad &\bar{R}  \\
		\text{s.t.} \quad & \eqref{eq:R1}, \eqref{eq:pmax1}, \eqref{eq:fn1} \nonumber \\
		& \bar{R}_{g}^{p} - \frac{B}{M} \sum_{m=1}^{M} \underset{u_{g,k}^p(m),\rho_{g,k}^p(m)}{\text{max}} \left( \frac{\log(\rho_{g,k}^p(m))-\rho_{g,k}^p(m) e_{g,k}^p(m) + 1}{\log(2)} \right) \leq 0, \nonumber\\
		& \qquad\qquad\qquad\qquad\qquad\qquad\qquad\qquad\qquad\qquad\qquad\qquad \forall k \in \mathcal{G}_g, \forall g \in \mathcal{G}, \label{eq:Rp3}\\	
		& \bar{R}_{g}^{c} - \frac{B}{M} \sum_{m=1}^{M} \underset{u_{g,k}^c(m),\rho_{g,k}^c(m)}{\text{max}} \left( \frac{\log(\rho_{g,k}^c(m))-\rho_{g,k}^c(m) e_{g,k}^c(m) + 1}{\log(2)} \right) \leq 0, \nonumber\\
		& \qquad\qquad\qquad\qquad\qquad\qquad\qquad\qquad\qquad\qquad\qquad\qquad \forall k \in \mathcal{M}_g, \forall g \in \mathcal{G}. \label{eq:Rc3}
	\end{align}
\end{subequations}  
Here $\mathcal{V}_2 \triangleq \left\lbrace \mathbf{w}_{g}^p, \mathbf{w}_{g}^c, \bar{R}_{g}^{p}, \bar{R}_{g}^{c}, {\bm \rho}_g^p,{\bm \rho}_g^c,\mathbf{u}_g^p,\mathbf{u}_g^c|\,\, \forall g \in \mathcal{G} \right\rbrace$ is the set of optimization variables. We explicitly introduce the MSE weights ${\bm \rho}_g^p = {\rm vec}\left(\left\{ \rho_{g,k}^p(m) | \,\, \forall k \in \mathcal{G}_g,\forall g \in \mathcal{G}, {\forall m: \, 1\leq m \leq M} \right\}\right)$ and ${\bm \rho}_g^c = {\rm vec}\left(\left\{ \rho_{g,k}^c(m) |\,\, \forall k \in \mathcal{M}_g, \forall g \in \mathcal{G}, \forall m: \, 1\leq m \leq M \right\}\right)$. Additionally, we introduce the receiver coefficients $\mathbf{u}_g^p = {\rm vec}\left(\left\{ u_{g,k}^p(m) | \,\,\forall k \in \mathcal{G}_g, \forall g \in \mathcal{G}, \forall m: \, 1\leq m \leq M \right\}\right)$ as well as $\mathbf{u}_g^c = {\rm vec}\left(\left\{ u_{g,k}^c(m) |\,\, \forall k \in \mathcal{M}_g, \forall g \in \mathcal{G}, \forall m: \, 1\leq m \leq M \right\}\right)$.
The constraints \eqref{eq:Rp3} and \eqref{eq:Rc3} explicitly show the channel realization dependence of these variables.\\
\indent Please note that problem \eqref{eq:Opt2} is convex if we fix the newly introduced variables $\{ {\bm \rho}_g^p,{\bm \rho}_g^c,\mathbf{u}_g^p,\mathbf{u}_g^c \}$. Also, these variables can be computed for fixed $\{ \mathbf{w}_{g}^p, \mathbf{w}_{g}^c, \bar{R}_{g}^{p}, \bar{R}_{g}^{c} \}$ using \eqref{eq:upmmse}, \eqref{eq:ucmmse}, as well as $\rho_{g,k}^p = 1/e_{g,k,\text{mmse}}^p$ and $\rho_{i,k}^c = 1/e_{i,k,\text{mmse}}^c$. Another advantage of the WMMSE formulation is that we can find a formulation for the optimization problem, which does not depend on the respective channel realization. To achieve such a reformulation, we first define the following set of variables for the WMMSE-based algorithm:
\begin{align}
%	u_{g,k}^p(m) &= (\mathbf{w}_g^p)^H \mathbf{h}_k^m / T_{g,k}^p &&u_{i,k}^c(m) = (\mathbf{w}_i^c)^H \mathbf{h}_k^m / T_{i,k}^c \label{eq:u} \\
%	\rho_{g,k}^p(m) &= \frac{1}{e_{g,k,\text{mmse}}^p(m)} \quad \rho_{i,k}^c(m) = \frac{1}{e_{i,k,\text{mmse}}^c(m)} \label{eq:rho} \\
	\overline{t}_{g,k}^o &= \frac{1}{M} \sum_{m=1}^{M} \rho_{g,k}^o(m) \left|u_{g,k}^{o}(m)\right|^2, &&\; o \in \{p,c\}, \label{eq:t}\\
	\overline{z}_{g,k}^o &= \frac{1}{M} \sum_{m=1}^{M} \left( 1 - \rho_{g,k}^o(m) + \log( \rho_{g,k}^o(m) ) \right), &&\; o \in \{p,c\}, \label{eq:z}\\
	\overline{\mathbf{f}}_{g,k}^o &= \frac{1}{M} \sum_{m=1}^{M} \rho_{g,k}^o(m) \mathbf{h}_k^m (u_{g,k}^o(m))^H, &&\; o \in \{p,c\}, \label{eq:f} \\
	\overline{\mathbf{Y}}_{g,k}^o &= \frac{1}{M} \sum_{m=1}^{M} \left( \rho_{g,k}^o(m)\left|u_{g,k}^{o}(m)\right|^2 \mathbf{h}_k^m (\mathbf{h}_k^m)^H \right), &&\; o \in \{p,c\}. \label{eq:Y}
\end{align}
These auxiliary variables are fixed during optimization, but they are updated after each iteration. The resulting optimization problem utilizing the variables from \eqref{eq:t} - \eqref{eq:Y} is then given by
\begin{subequations}\label{eq:Opt3}
	\begin{align}
		\underset{\mathcal{V}_1}{\text{max}}\quad &\bar{R}  \\
		\text{s.t.} \quad & \eqref{eq:R1}, \eqref{eq:pmax1}, \eqref{eq:fn1}, \nonumber \\
		& \sum_{j\in\mathcal{G}} (\mathbf{w}_j^p)^H \overline{\mathbf{Y}}_{g,k}^p \mathbf{w}_j^p + \sum_{l\in\Omega_k} (\mathbf{w}_l^c)^H \overline{\mathbf{Y}}_{g,k}^p \mathbf{w}_l^c - 2 \text{Re}\{ (\overline{\mathbf{f}}_{g,k}^p)^H \mathbf{w}_g^p \} \nonumber\\ 
		& \qquad \qquad + \frac{\log(2)}{B} \bar{R}_{g}^{p} + \sigma^2 \overline{t}_{g,k}^p - \overline{z}_{g,k}^p \leq 0, \qquad\qquad\qquad\qquad \forall k \in \mathcal{G}_g, \forall g \in \mathcal{G}, \\	
		& \sum_{j\in\mathcal{G}} (\mathbf{w}_j^p)^H \overline{\mathbf{Y}}_{g,k}^c \mathbf{w}_j^p + \sum_{l\in\Omega_k} (\mathbf{w}_l^c)^H \overline{\mathbf{Y}}_{g,k}^c \mathbf{w}_l^c + \sum_{m\in\Psi_{g,k}} (\mathbf{w}_m^c)^H \overline{\mathbf{Y}}_{g,k}^c \mathbf{w}_m^c \nonumber\\ 
		& \qquad \qquad  + (\mathbf{w}_g^c)^H \overline{\mathbf{Y}}_{g,k}^c \mathbf{w}_g^c - 2 \text{Re}\{ (\overline{\mathbf{f}}_{g,k}^c)^H \mathbf{w}_g^c \} + \frac{\log(2)}{B} \bar{R}_{g}^{c} \nonumber \\
		& \qquad \qquad + \sigma^2 \overline{t}_{g,k}^c - \overline{z}_{g,k}^c \leq 0, \qquad\qquad\qquad\qquad\qquad\qquad\quad \forall k \in \mathcal{M}_g, \forall g \in \mathcal{G}. 
	\end{align}
\end{subequations}  
Note that the set of optimization variables is now reduced to the set $\mathcal{V}_1$. Problem \eqref{eq:Opt3} is a convex optimization problem and thus amenable for solving using standard tools (e.g., CVX \cite{cvx}). %Although the variables from \eqref{eq:t} - \eqref{eq:Y} are dependent on the sample size $M$, they are not subject to optimization. 
An advantage of problem \eqref{eq:Opt3} is the independence of the Monte-Carlo sample size $M$, which pronounces the computational complexity advantage of this solution compared to state-of-the-art algorithms. Typically, $M$ has to be very large in order for the SAA to be an accurate approximation. In order to find a stationary solution to \eqref{eq:Opt3}, Algorithm \ref{alg} lists detailed steps for the WMMSE-based optimization to maximize the minimum achievable rate of all mutlicast groups.
\begin{algorithm}[h]
	\caption{WMMSE-based Minimum Achievable Rate Maximization}
	\begin{spacing}{1.1}
	\begin{algorithmic}[1]
		\STATE Initialize $\mathbf{w}_g^p$ and $\mathbf{w}_g^c$ to feasible values $\forall g\in\mathcal{G}$. Generate $M$ channel vector samples as $\{\mathbf{h}^1,\cdots,\mathbf{h}^M\}$
		\STATE \textbf{Repeat:} until convergence 
		\STATE Update the set of auxiliary variables $\{ \overline{t}_{g,k}^o, \overline{z}_{g,k}^o, \overline{\mathbf{f}}_{g,k}^o, \overline{\mathbf{Y}}_{g,k}^o | o\in\{ p,c \} \}$
		\STATE Solve optimization problem \eqref{eq:Opt3}
		\STATE \textbf{End}  
		%\EndWhile
	\end{algorithmic}
	\end{spacing}
	\label{alg}
\end{algorithm}%
\begin{theorem} \label{KKT}
	The solution $\Psi_\nu = \{ \mathbf{w}_{g}^o, \bar{R}_{g}^{o}, {\bm \rho}_g^o,\mathbf{u}_g^o,|\,\, \forall g \in \mathcal{G}, \forall o \in\{p,c\} \} $ is generated by Algorithm \ref{alg} in iteration $\nu$. The sequence $\{ \Psi_\nu \}_{\nu=1}^{\infty}$ converges to a KKT point of problem \eqref{eq:Opt2}.
\end{theorem}
\begin{proof}
	Please refer to Appendix \ref{app2}.
\end{proof}
\subsection{Complexity Analysis}
%\textcolor{red}{We also need some discussion of the receiver complexity. How many SIC is needed at each user. How the clustering algorithm control the number of SIC at each user.}\\
In this subsection, we are interested in the overall computational complexity of Algorithm \ref{alg}. This overall complexity mainly depends on two factors: $1)$ the complexity of problem \eqref{eq:Opt3}; $2)$ the solution accuracy of the iterative procedure. Update step $3$ of Algorithm \ref{alg}, which essentially updates the beamforming vectors and the allocated rates, solves the quadratic constrained convex optimization problem (QCCP) \eqref{eq:Opt3}. In order to state the worst-case computational complexity for Algorithm \ref{alg}, we first need to determine the number of constraints and variables. There are $d_1 = 2N+G+\sum_{g\in\mathcal{G}}|\mathcal{G}_g|+\sum_{g\in\mathcal{G}}|\mathcal{M}_g|$ constraints, with $|\mathcal{M}_g|$ the cardinality of $\mathcal{M}_g$, and $d_2 = (2G(NL+1))$ variables. The complexity metric of problem \eqref{eq:Opt3} becomes $\mathcal{O}((d_1 d_2^2 + d_2^3)\sqrt{d_1}\log(1/\epsilon))$, for given solution accuracy $\epsilon$. %The mayor advantage of such a WMMSE-based solution approach, is the independence of the Monte-Carlo sample size $M$. Typically, $M$ has to be very large in order for the SAA to be an accurate approximation. 
\section{Numerical Simulations} \label{sec:NS}
In this section, numerical simulations illustrating the performance of our proposed scheme, compared to TIN and the SCM-RSMA scheme, are conducted. We consider a cache-enabled C-RAN over a square area of [-400 400] $\times$ [-400 400] $\text{m}^2$. Unless specified otherwise, we set the maximum transmit power to $P_n^{\text{max}} = 28$ dBm and the total number of files in the library to $F=50$. The noise spectral density is assumed to be $-168$ dBm/Hz. Also, we set the Monte-Carlo sample size to $M=1000$ and the parameter for Algorithm~\ref{alg_clust} to $A_n^{\text{max}}=8,\forall n\in\mathcal{N}$. The channel model used for our simulations is given as follows \cite{6786060}:
\begin{equation}
	\mathbf{h}_{n,k} = 	D_{n,k}\mathbf{e}_{n,k},
\end{equation}
with
\begin{equation}
	D_{n,k} = 10^{-\text{PL}_{n,k}/20}\sqrt{g_{n,k}s_{n,k}}.
\end{equation}
Here $g_{n,k}$ is the shadowing coefficient, $s_{n,k}$ is the antenna gain, and $\text{PL}_{n,k}$ defines the path-loss between BS $n$ and user $k$ given by
\begin{equation}
	\text{PL}_{n,k} = 148.1 + 37.6\log_{10}(d_{n,k}),
\end{equation}
where $d_{b,k}$ is the distance in km between BS $b$ and user $k$. The small-fading coefficients $\mathbf{e}_{n,k}\in\mathbb{C}^{L\times 1}$ are modeled as $\mathbf{e}_{n,k}\sim \mathbf{\mathcal{CN}}\left(0, \mathbf{I}_{L}\right)$. This paper distinguishes between two cases: $1$) full CSIT, where the CP has perfect channel knowledge, i.e., the vector $\mathbf{h}_{n,k}$ is known; $2$) statistical CSIT, where the CP can only perfectly estimate the large-fading coefficient $D_{n,k}$, whereas the small-fading coefficient $\mathbf{e}_{n,k}$ is unknown. \\
\indent As baseline schemes, we deploy two different approaches contrary to the proposed RS-CMD scheme. As the most basic scheme, TIN is used allowing only private streams to be transmitted. More specific, the data for each multicast group is encoded into a private stream and then transmitted by the respective serving BSs. As every user in a multicast group only decodes their group's private message, the rest of the streams are treated as noise. The second baseline scheme, referred to as SCM-RSMA, is based on the concept of \cite{7555358}. In \cite{7555358}, the authors propose an RSMA scheme using one common stream, which is then decoded by all users in the network. More specific, there are $G+1$ streams to be transmitted, one for every multicast group, i.e., the private streams, and one mutual stream to be decoded by all groups, i.e., the common stream.
\subsection{MMF Rate vs. Fronhaul Capacity}\label{mmrvfc}
\begin{figure}
	\centering
	\begin{subfigure}[t]{0.49\textwidth}
		\centering
		\includegraphics[width=1\linewidth]{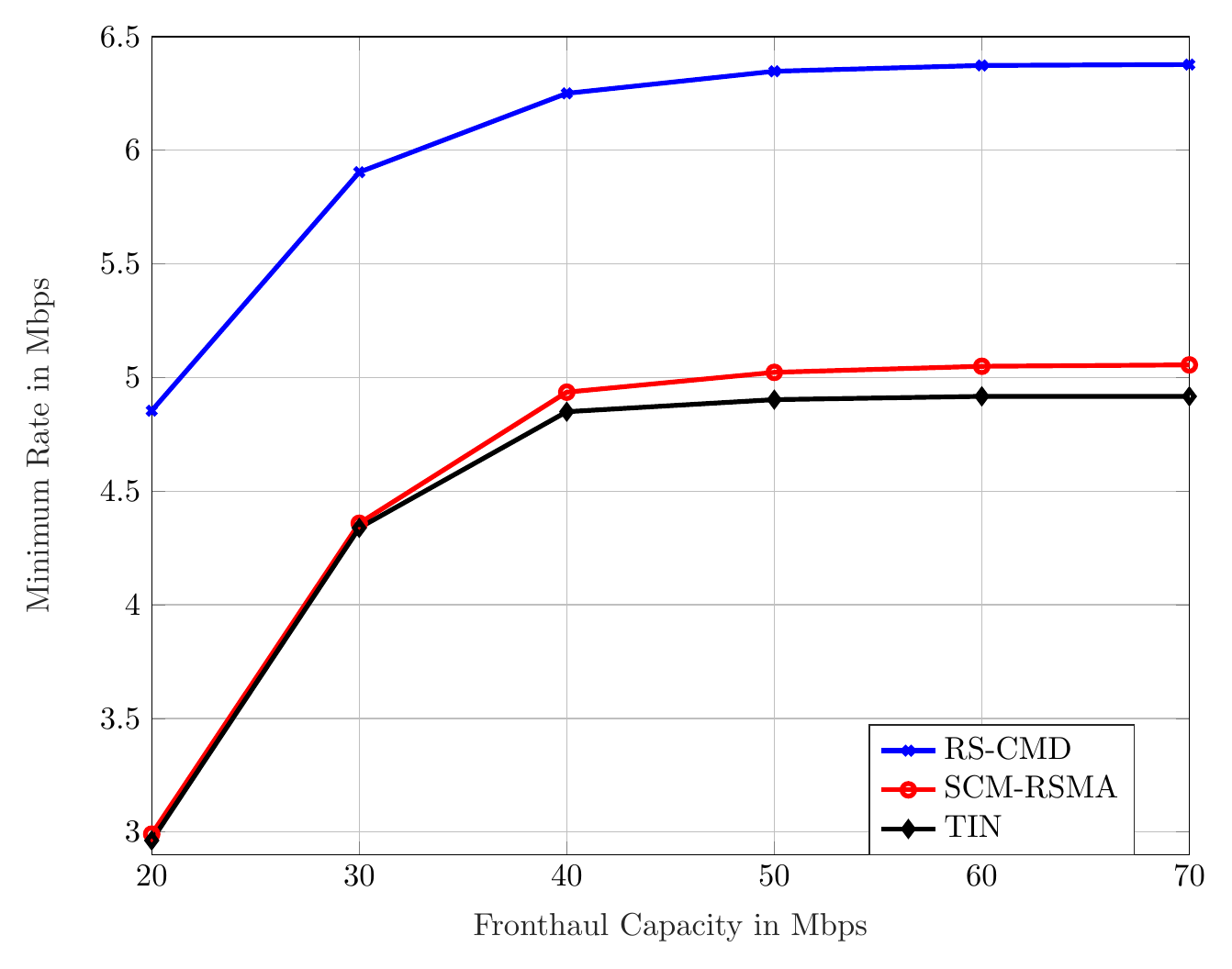}
		\caption{MMF rate vs. fronthaul capacity for statistical CSIT.}
		\label{y_rate_x_fthl_im}
	\end{subfigure}\hfill
	\begin{subfigure}[t]{0.49\textwidth}
		\centering
		\includegraphics[width=1\linewidth]{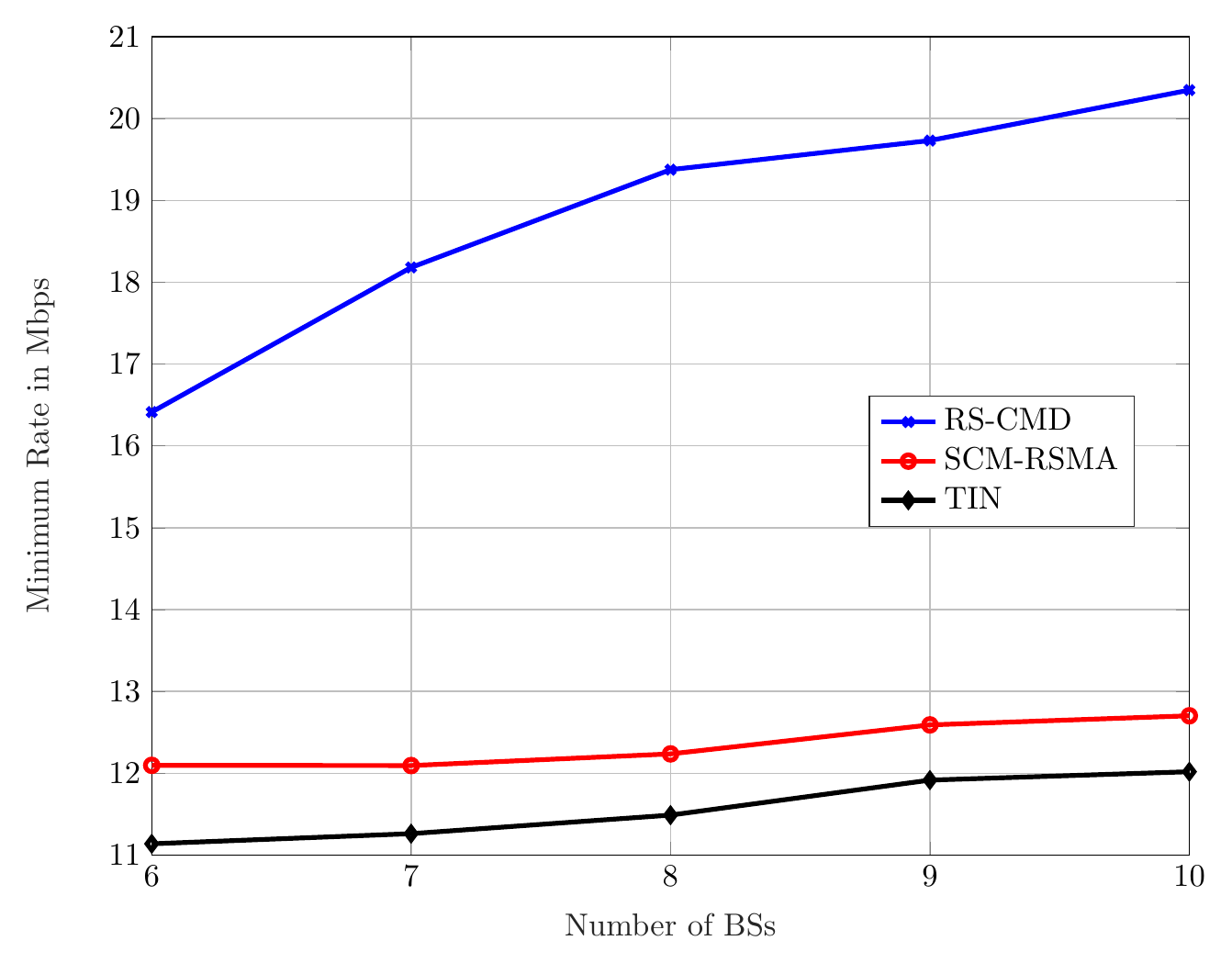}
		\caption{MMF rate vs. number of BSs for full CSIT.}
		\label{y_rate_x_nuBS_im}
	\end{subfigure}
	\caption{MMF rate as a function of different parameters. We consider the proposed RS-CMD scheme, the reference SCM-RSMA, and TIN.}
	\label{y_rate_x_diff_im}
\end{figure}
In this setup, we consider a C-RAN consisting of $7$ BSs and $15$ users that are randomly placed throughout the network with uniform distribution. The BSs are equipped with $L=2$ transmit antennas. The cache memory size is considered to be $5$ files at each BS ($10$\% of the total files). \\
\indent Fig.~\ref{y_rate_x_fthl_im} shows the performance of all studied schemes as a function of the fronthaul capacity in Mbps for statistical CSIT. 
%\textcolor{blue}{Contrary to that, Fig.~\ref{y_rate_x_fthl_fcsit_im} shows the same relationship for full CSIT. However, in Fig.~\ref{y_rate_x_fthl_fcsit_im}, the fronthaul capacity under consideration spans from $40$ Mbps to $140$ Mbps. As expected, we can observe higher rates as the full CSIT is available. Different to Fig.~\ref{y_rate_x_fthl_im}, in Fig.~\ref{y_rate_x_fthl_fcsit_im}, the rates of all schemes suffer no saturation at higher fronthaul capacities.}
A general observation from Fig.~\ref{y_rate_x_fthl_im} is that our proposed RSMA scheme attains significant gain in terms of the achievable minimum rate, especially in the low fronthaul capacity regime. Specifically, in Fig.~\ref{y_rate_x_fthl_im}, the gain compared to conventional TIN is 63.85\% when the fronthaul capacity per-BS is equal to 20 Mbps and 29.69\% when the fronthaul capacity is increased to 70 Mbps. Through mitigating the interference more efficiently using RS-CMD, our proposed scheme achieves higher minimum rates compared to both reference schemes. Moreover, when the fronthaul capacity is $30$ Mbps and higher, SCM-RSMA performs better than the conventional TIN scheme. These results clearly highlight the need for utilizing RSMA schemes in order to achieve a higher minimum rate. That is, RSMA assists the network to better ensure fairness among all users, especially using the proposed RS-CMD scheme. % Warum hamdi so viel schlechter...
%\begin{figure}
%	\centering
%	\includegraphics[width=.49\linewidth]{figures/y_rate_x_fthl.pdf}
%	\caption{Max-min rate as a function of the fronthaul capacity for the statistical CSIT case. We consider the proposed RSMA-CMD scheme, the reference RSMA scheme 1, and TIN.}
%	\label{y_rate_x_fthl_im}
%%	\vspace{-0.7cm}
%\end{figure}
\subsection{Minimum Achievable Rate Under Full CSIT}
For the following set of simulations, we design a network with $15$ users, a fronthaul capacity of $80$ Mbps, $5$ files cache size, and $L=2$ antennas at each BS. The number of BSs controlled by the CP is the parameter under study, we vary the total BSs from $6$ to $10$. \\
\indent The minimum achievable rate is plotted as a function of the number of users in Fig.~\ref{y_rate_x_nuBS_im}. Consistent with the results from subsection \ref{mmrvfc}, our proposed scheme outperforms TIN and SCM-RSMA significantly. This gain is especially substantial when the number of BSs is high.\\
\indent In constrast to the statistical CSIT case, in Fig.~\ref{y_rate_x_fthl_im}, we observe greater gains of SCM-RSMA over TIN, as well as greater gains of our proposed RS-CMD scheme over both baselines. Full knowledge of CSIT obviously results in better rates and thus better performance of all schemes.
%\begin{figure}
%	\centering
%	\includegraphics[width=.49\linewidth]{figures/y_rate_x_nuBS.pdf}
%	\caption{Max-min rate as a function of ... for the full CSIT case. We consider the proposed RSMA-CMD scheme, the reference RSMA scheme 1, and TIN.}
%	\label{y_rate_x_nuBS_im}
%	%	\vspace{-0.7cm}
%\end{figure}
\subsection{Impact of Caching on the MMF Rate}
In another set of simulations, we focus on the impact of different cache sizes on the achievable minimum rate using the RS-CMD scheme in the statistical CSIT scenario. Except for the different cache sizes, the network under consideration does not differ from the subsection \ref{mmrvfc}. \\
\indent In Fig.~\ref{y_rate_x_fthl_cach_im}, we fix the cache size to different values, i.e., $\{0,5,\cdots,25\}$ files, and plot the achievable minimum rate for different fronthaul capacities. Matching with previous results, the minimum achievable rate increases with fronthaul capacity. First, in Fig.~\ref{y_rate_x_fthl_cach_im}, we observe the gap between the graphs, i.e., the caching gain. Especially in low fronthaul regimes, the MMF rate is better when bigger caches are utilized, since the fronthaul capacity is limited and thus transmitting data from the CP to the BSs is a sensible operation. A cache hit reduces the congestion on the fronthaul links. \\
\indent Another major observation in Fig.~\ref{y_rate_x_fthl_cach_im} is the decrease in caching gain with increasing fronthaul capacity, i.e., the gap between the graphs is decreasing. When a high quantity of fronthaul capacity is available, the importance of using caching lowers. As the fronthaul links are barely congested, there is no need to alleviate the congestion with cache hits.
Since low fronthaul regimes are more common in practice, the results herein pronounce the role of caching in the proposed system model as a means to increase the MMF rate among the multicast groups. \\
\indent In Fig.~\ref{y_rate_x_fthl_cach_comp_im}, we show the MMF rate of RS-CMD, SCM-RSMA, and TIN as a function of the fronthaul capacity. A comparison is conducted for three different cache sizes, i.e., $\{0,20,40\}$ files per BS ($\{0,10,20\}$\% of the total files). Generally, RS-CMD outperforms TIN and SCM-RSMA. Interestingly, the performance gain of RS-CMD over the benchmarks changes with different cache sizes. At $25$ Mbps fronthaul capacity, RS-CMD gains $18.25$\%, $33.79$\%, and $67.08$\% over SCM-RSMA with $20$, $10$, and $0$ files in the cache, respectively. In contrast, at $40$ Mbps fronthaul capacity, RS-CMD gains $16.28$\%, $16.99$\%, and $25.84$\% over SCM-RSMA with $20$, $10$, and $0$ files in the cache, respectively. RS-CMD significantly outperforms TIN and SCM-RSMA at low fronthaul regimes with small cache sizes. 
% Increase = New Number - Original Number
% % increase = Increase ÷ Original Number × 100.
% 25Mbps: 40%: ( 6.22 - 5.26 ) / 5.26 * 100 = 18.25
%		  20%: ( 5.90 - 4.41 ) / 4.41 * 100 = 33.79
%		   0%: ( 5.38 - 3.22 ) / 3.22 * 100 = 67.08
% 40Mbps: 40%: ( 6.57 - 5.65 ) / 5.65 * 100 = 16.28
%		  20%: ( 6.47 - 5.53 ) / 5.53 * 100 = 16.99
%		   0%: ( 6.33 - 5.03 ) / 5.03 * 100 = 25.84
%\begin{figure}
%	\centering
%	\includegraphics[width=.49\linewidth]{figures/y_rate_x_fthl_cach.pdf}
%	\caption{Max-min rate as a function of the fronthaul capacity for the statistical CSIT case. We consider different cache sizes at the BSs.}
%	\label{y_rate_x_fthl_cach_im}
%	%	\vspace{-0.7cm}
%\end{figure}
\begin{figure}
	\centering
	\begin{subfigure}[t]{0.49\textwidth}
		\centering
		\includegraphics[width=1\linewidth]{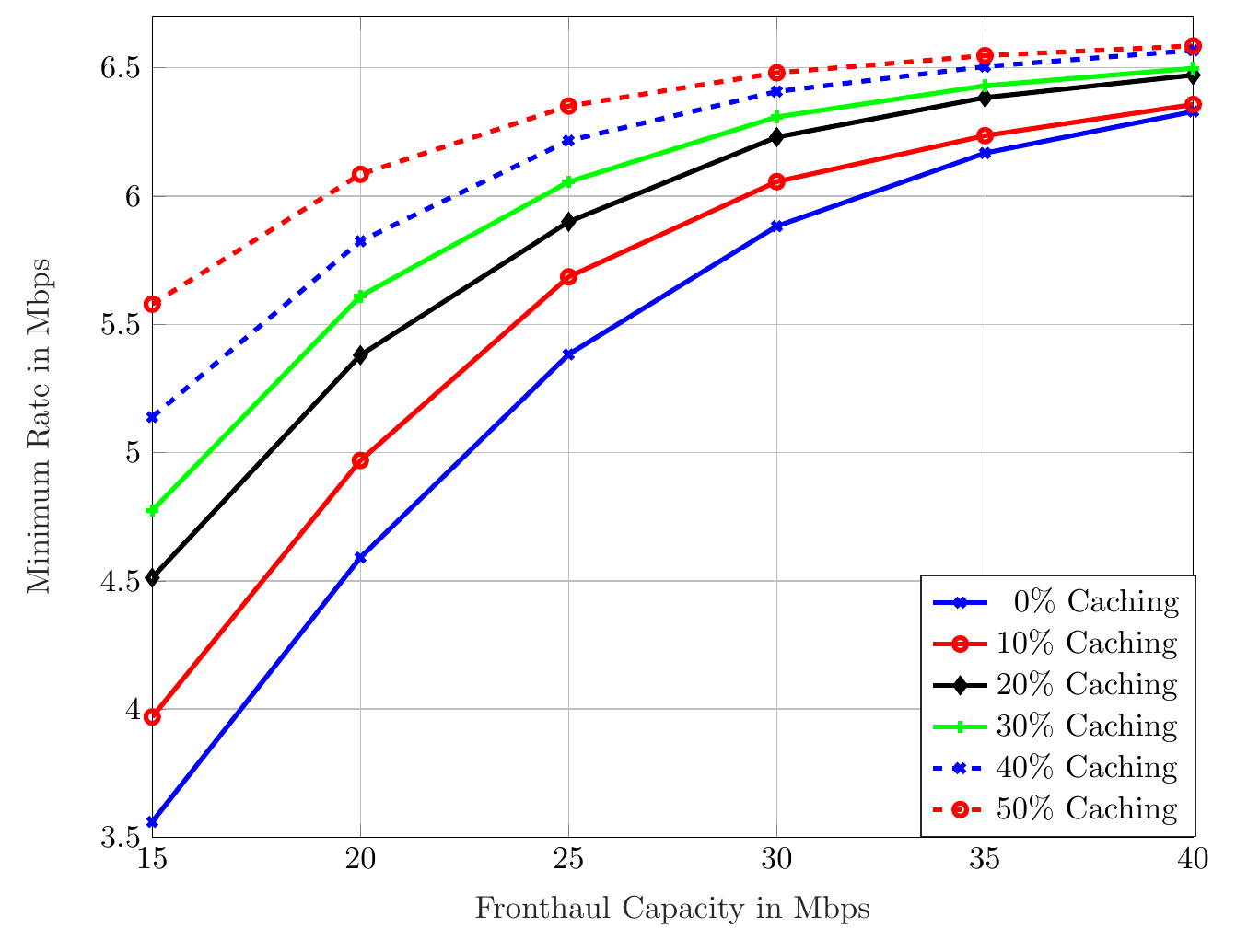}
		\caption{Considering RS-CMD only.}
		\label{y_rate_x_fthl_cach_im}
	\end{subfigure}\hfill
	\begin{subfigure}[t]{0.49\textwidth}
		\centering
		\includegraphics[width=1\linewidth]{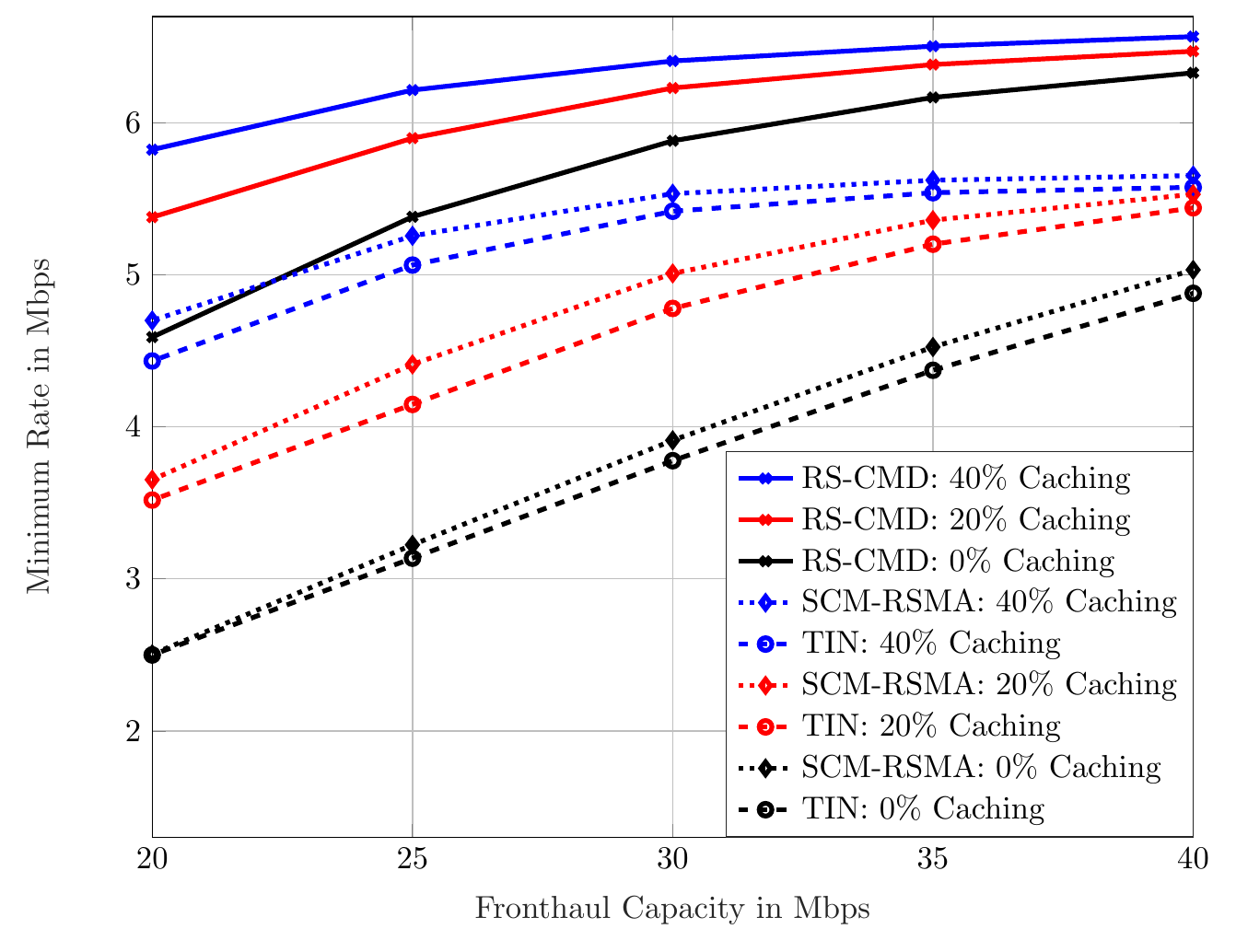}
		\caption{Comparison between the schemes.}
		\label{y_rate_x_fthl_cach_comp_im}
		%	\vspace{-0.7cm}
	\end{subfigure}
	\caption{MMF rate as a function of the fronthaul capacity for the statistical CSIT case. We consider the impact of different cache sizes.}
	\label{y_rate_x_fthl_para_im}
\end{figure}

\subsection{Influence of Transmit Antennas on the MMF Rate}
In order to investigate the influence of different amounts of transmit antennas at the BSs on the MMF rate among all multicast groups, we conduct the next set of simulations. We consider the same network as in subsection \ref{mmrvfc}. \\ %a network consisting of $7$ BSs, $15$ users, and a cache memory of $5$ files ($10$\% of the total files). 
\indent The minimum achievable rate as a function of the fronthaul capacity for $L \in \{1,\cdots,5\}$ antennas is shown in Fig.~\ref{y_rate_x_fthl_numAnt_im}. As expected, we observe higher rates when utilizing more transmit antennas at the BSs. Interestingly, especially in high fronthaul capacity regimes, the minimum rate of the single-antenna system is considerably lower than the MMF rate in all other networks. This highlights the need for utilizing multiple antennas in the proposed cache-aided C-RAN with the multigroup multicast transmission scheme.

Another observation in Fig.~\ref{y_rate_x_fthl_numAnt_im} is the similarity of the minimum achievable rate associated with $4$ versus $5$ antennas per BS systems. Throughout all fronthaul capacities, the rate of the $5$ antenna system achieves only minor gains over the $4$ antenna system. Therefore, considering the herein defined network parameters, using more than $4$ antennas to increase the minimum achievable rate is not a viable option. This is partially caused by power constraints at the BSs and partially by the fronthaul capacity constraints.
\begin{figure}
	\centering
	\includegraphics[width=.49\linewidth]{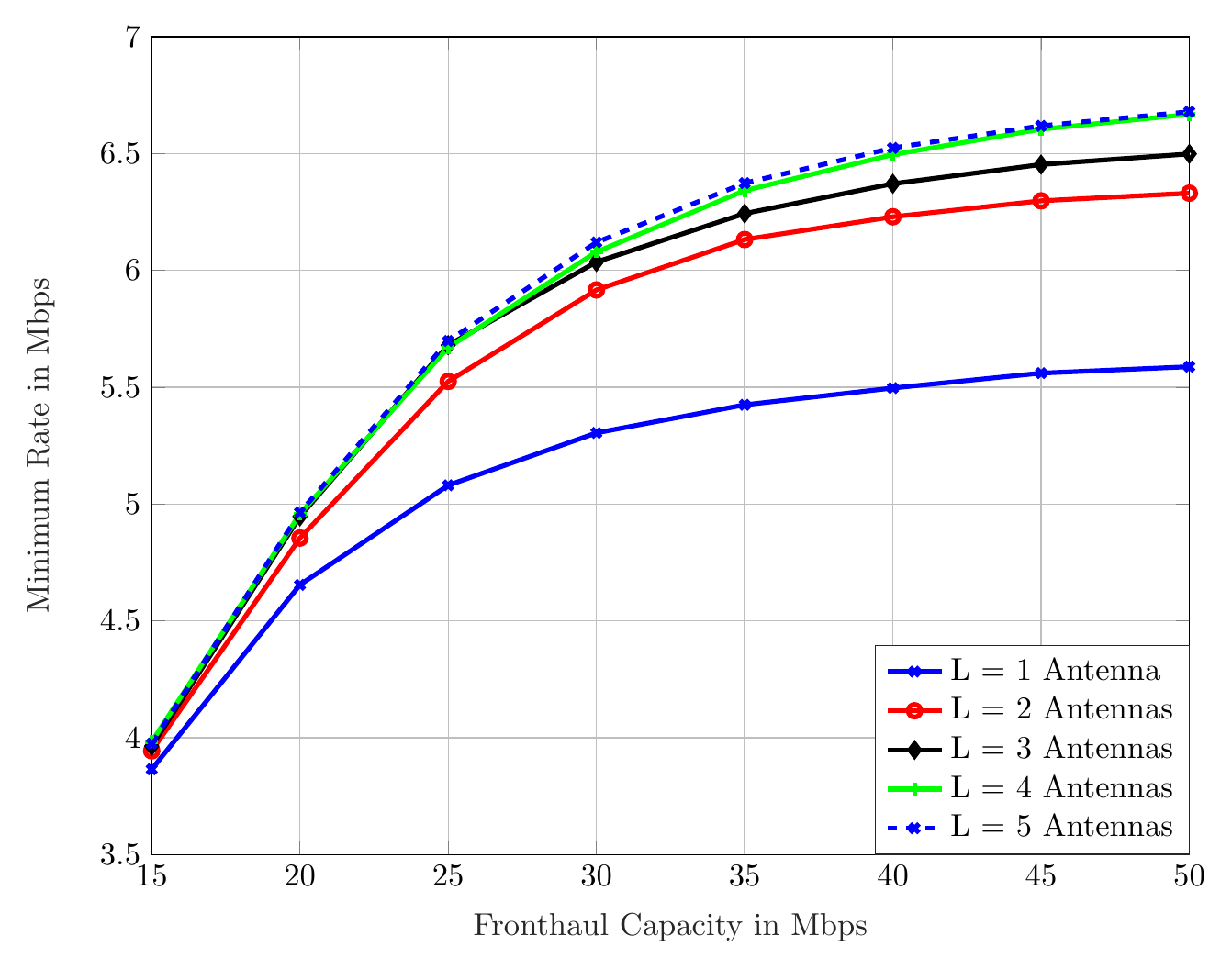}
	\caption{MMF rate as a function of the fronthaul capacity for the statistical CSIT case. We consider different numbers of transmit antennas at the BSs.}
	\label{y_rate_x_fthl_numAnt_im}
	%	\vspace{-0.7cm}
\end{figure}

\subsection{Comparison to a Simpler Transmission Scheme}\label{mul_uni}
Prior to introducing the next set of simulations, we present another state-of-the-art transmission scheme to be used as a performance reference. In this work, we group users into multicast groups based on their requested content, i.e., all users in a group receive the same private and common message. Many related works, e.g., \cite{7925639,7499119}, consider a special case of this scheme. More specific, following the notations of this paper, each group comprises only one user. This is referred to as RS-CMD with $G= K$, while our proposed scheme is referred to as RS-CMD with $G\leq K$. \\
\indent Now, comparing our proposed multigroup multicast transmission scheme to the conventional scheme, in Fig.~\ref{y_rate_x_numU_unimul_v2_im}, we show the minimum achievable rate of all groups as a function of the number of users. The fronthaul capacity is fixed to $30$ Mbps, we assume $7$ BSs throughout the network and a cache memory size of $10$ files ($20$\% of the total files). \\
\indent Generally, in Fig.~\ref{y_rate_x_numU_unimul_v2_im}, RS-CMD with $G\leq K$ outperforms RS-CMD with $G= K$. The performance gap is small when the number of users is low, since multicast groups are more unlikely in this regime as the user requests are not coinciding with high probability. Note that RS-CMD with $G\leq K$ and with $G= K$ are equal if no users request the same content. \\
\indent As expected, the MMF rate among all groups decreases with the number of users for both schemes. Since the BSs serve additional users under the same power and fronthaul constraints, integrating more users into the same network reduces the MMF rate amongst the groups. Under RS-CMD with $G\leq K$, these additional users may be included in an existing multicast group if they request the same content. Therefore, the proposed scheme suffers less MMF rate decrease as compared to RS-CMD with $G=K$, i.e., see Table~\ref{y_rate_x_numU_unimul_v2_tb}.
%In case no users request the same content, our proposed scheme is equal to the RS-CMD with $G= K$ transmission. With increasing numbers of users, the gain of our proposed scheme against the reference in terms of the minimum achievable rate increases, see Table~\ref{y_rate_x_numU_unimul_v2_tb}. Since all users in a multicast group receive the same messages, they can, to an extend, be regarded as a single user. As previously mentioned, networks with less users typically achieve higher minimum rates, therefore these results pronounce the superiority of the multicast transmission scheme when studying the minimum achievable rate.}
%\begin{figure}
%	\centering
%	\includegraphics[width=.49\linewidth]{figures/y_rate_x_numU_unimul_v2.pdf}
%	\caption{Max-min rate as a function of the number of uses for the imperfect CSIT case. We consider two different transmission schemes, i.e., multicasting and unicasting.}
%	\label{y_rate_x_numU_unimul_v2_im}
%	%	\vspace{-0.7cm}
%\end{figure}
% --- ... ---
\begin{figure}
	\centering
	\begin{subfigure}[t]{0.49\textwidth}
		\centering
		\includegraphics[width=1\linewidth,valign=t]{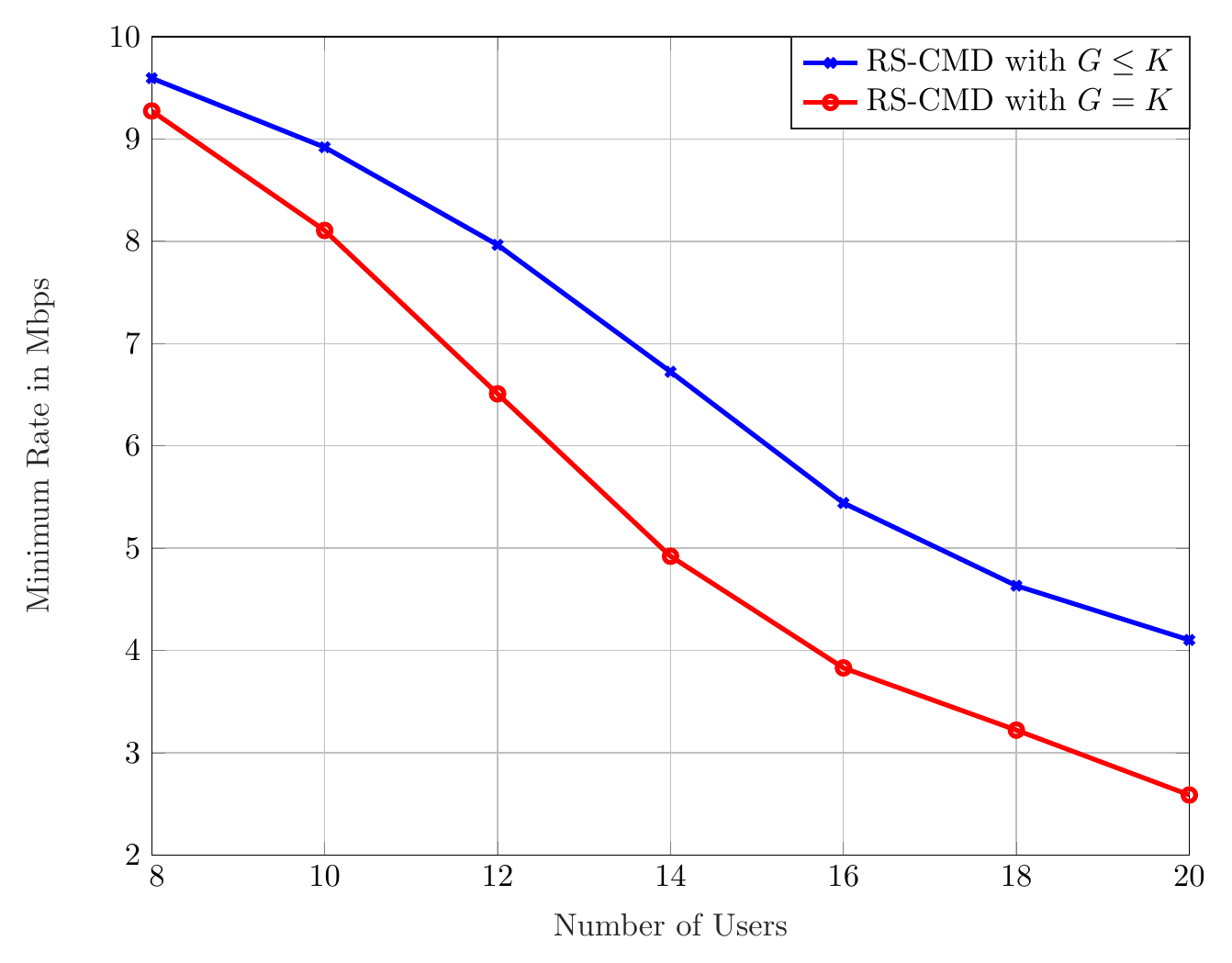}
		\caption{MMF rate of RS-CMD with $G\leq K$ and $G=K$ as a function of the number of uses.}
		\label{y_rate_x_numU_unimul_v2_im}
	\end{subfigure}\hfill
	\begin{subfigure}[t]{0.49\textwidth}
			\centering
			\adjustbox{valign=t}{\begin{tabular}{c|c}
%				\hline
				Number of Users & Multicast Gain in \% \\
				\hline
				$8$ & $3.45$ \\
				\hline
				$10$ & $10.03$ \\
				\hline
				$12$ & $22.37$ \\
				\hline
				$14$ & $36.64$ \\
				\hline
				$16$ & $42.10$ \\
				\hline
				$18$ & $43.83$ \\
				\hline
				$20$ & $58.55$ \\
%				\hline
			\end{tabular}}
		\vspace*{2.3cm}
		\caption{MMF-rate gain of RS-CMD with $G\leq K$ over $G=K$ in percent.}
		\label{y_rate_x_numU_unimul_v2_tb}
	\end{subfigure}
	\caption{MMF rate versus the number of uses for the statistical CSIT case. We consider two different transmission schemes, i.e., RS-CMD with $G\leq K$ and $G=K$.}
	\label{y_rate_x_numU}
\end{figure}
\subsection{Convergence Behaviour}
\begin{figure}
	\centering
	\includegraphics[width=.49\linewidth]{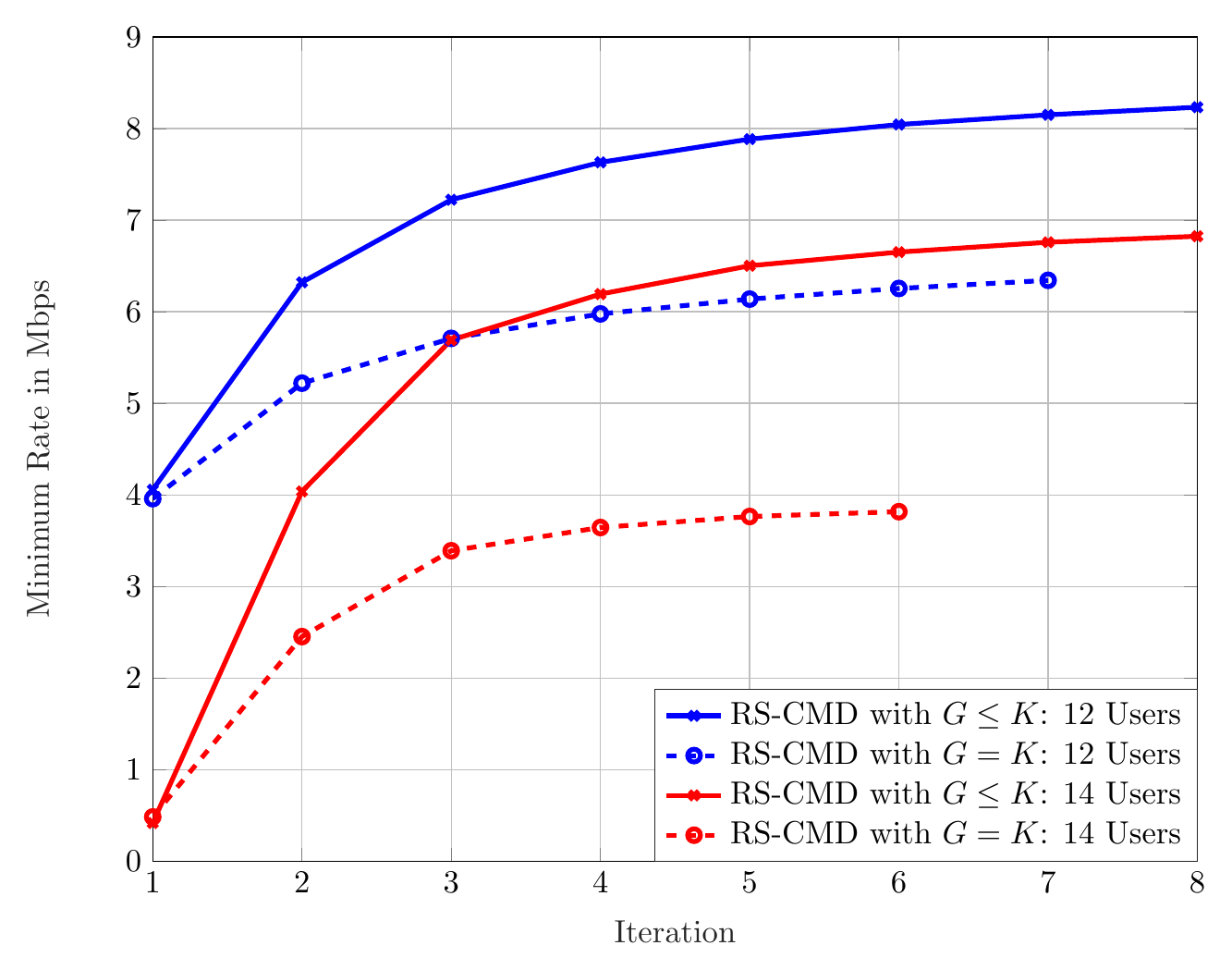}
	\caption{Convergence behavior of the proposed algorithm under statistical CSIT. Two different number of users are considered for RS-CMD with $G\leq K$ and $G=K$.} %Four different number of users are considered.}
	\label{y_rate_x_conv_im}
	%	\vspace{-0.7cm}
\end{figure}
To now illustrate the convergence behavior of our proposed Algorithm \ref{alg}, we consider the same simulation parameters as in section \ref{mul_uni}. Fig.~\ref{y_rate_x_conv_im} shows the minimum achievable rate for RS-CMD with $G\leq K$ and $G=K$ when the network consists of $12$ and $14$ users in every step of Algorithm \ref{alg} until convergence. Since the overall required iterations until convergence is relatively low, the results in Fig.~\ref{y_rate_x_conv_im} exhibit the feature of high execution speed of our proposed algorithm. Therefore, we observe good convergence behavior for the considered system model which further highlights another feature of the proposed algorithm.

\section{Conclusion} \label{sec:C}
Handling the sophisticated requirements of B5G wireless communication networks is a difficult task. To tackle this task, this paper proposed a promising RSMA-assisted cache-enabled C-RAN under a multigroup multicast scenario. The proposed framework considers a more practical scenario where the CP and BSs operate under statistical CSIT. We aimed at designing the precoders and BS clustering with the aim of maximizing the minimum achievable rate amongst the multicast groups. The optimization problem is tackled by a novel multicast group-based clustering approach, as well as an effective algorithm based on SAA and WMMSE. Extensive numerical simulations showed the significant max-min rate gain of the proposed RSMA framework over the existing benchmarks in cache-aided C-RAN with various fronthaul capacities and cache sizes. Therefore, we conclude that RSMA has a great potential to enhance user fairness and spectral efficiency in cache-aided C-RAN.
\appendices
\section{Proof of Theorem \ref{th1}} \label{app1}
Let the noise variance be non-zero, i.e., $\sigma^2 > 0$, the transmit power be finite, i.e., $P_n^{\text{max}} < \infty$, and the channel realizations be bounded. Then, the SINR expressions \eqref{eq:e2.18} and \eqref{eq:e2.19} are finite, i.e., $\gamma_{g,k}^p < \infty,$ $\forall g\in\mathcal{G},$ and $\gamma_{g,k}^c < \infty,$ $\forall k \in \mathcal{M}_g, \forall g \in \mathcal{G}$. The limit of average spectral efficiency exists, when the sample size $M$ tends to infinity. Let $\mathcal{W}$ be the feasible set of beamforming vectors determined by constraint \eqref{eq:pmax1}. All previous assumptions make sure, that $\mathcal{W}$ is compact and not empty. Now, considering the ergodicity assumption and the law of large numbers, we make the following statement \cite[Theorem 7.48]{saa}
\begin{equation}
	\underset{\mathbf{w}\in\mathcal{W}}{\text{sup}} \left| \frac{1}{M} \log_2(1+\gamma_{g,k}^o) - \mathbb{E}_\mathbf{h}\{ \log_2(1+\gamma_{g,k}^o) \} \right| \rightarrow 0, \quad \text{as } M \rightarrow \infty, \quad o\in\{p,c\}.
\end{equation}
Thus, given unlimited sample size, the SAA estimate of the rates converges to the ergodic rate uniformly on the compact set $\mathcal{W}$ with probability one. Therefore, the set of optimal solutions of problem \eqref{eq:Opt1} converges uniformly to the optimal solution set of problem \eqref{eq:Opt01} \cite[Theorem 5.3]{saa}. This completes the proof.
\section{Proof of Theorem \ref{KKT}} \label{app2}
The steps in this proof can be seen analog to \cite[Theorem 2]{6190228}. Let $G_g(\mathbf{w}_{g}^p, \mathbf{w}_{g}^c, {\bm \rho}_g^p,{\bm \rho}_g^c,\mathbf{u}_g^p,\mathbf{u}_g^c) = G_g^p(\mathbf{w}_{g}^p, {\bm \rho}_g^p,\mathbf{u}_g^p) + G_g^c(\mathbf{w}_{g}^c, {\bm \rho}_g^c,\mathbf{u}_g^c)$, where we define the following two functions
\begin{align}
	G_g^p(\mathbf{w}_{g}^p, {\bm \rho}_g^p,\mathbf{u}_g^p) &= \frac{B}{M\log(2)} \sum_{m=1}^{M} \underset{u_{g,k}^p(m),\rho_{g,k}^p(m)}{\text{max}} \left( \log(\rho_{g,k}^p(m))-\rho_{g,k}^p(m) e_{g,k}^p(m) + 1 \right),\\
	G_g^c(\mathbf{w}_{g}^c, {\bm \rho}_g^c,\mathbf{u}_g^c) &= \frac{B}{M\log(2)} \sum_{m=1}^{M} \underset{k\in\mathcal{M}_g}{\text{min}} \left( \underset{u_{g,k}^c(m),\rho_{g,k}^c(m)}{\text{max}} \left( \log(\rho_{g,k}^c(m))-\rho_{g,k}^c(m) e_{g,k}^c(m) + 1 \right) \right).
\end{align}
Using these definitions, we formulate following problem
\begin{subequations}\label{eq:OptP}
	\begin{align}
		\underset{\mathcal{V}_3}{\text{max}}\quad &\sum_{g\in\mathcal{G}} G_g(\mathbf{w}_{g}^p, \mathbf{w}_{g}^c, {\bm \rho}_g^p,{\bm \rho}_g^c,\mathbf{u}_g^p,\mathbf{u}_g^c)  \\
		\text{s.t.} \quad & \eqref{eq:pmax1}, \nonumber \\
		& \sum_{{g \in \mathcal{G}_n^p}}G_g^p(\mathbf{w}_g^p,\mathbf{u}_g^p,{\bm \rho}_g^p)  + \sum_{{g \in \mathcal{G}_n^c}}G_g^c(\mathbf{w}_g^c,\mathbf{u}_g^c,{\bm \rho}_g^c) \leq F_{n}, &&\forall n \in \mathcal{N}.
	\end{align}
\end{subequations}  
Here $\mathcal{V}_3 \triangleq \left\lbrace \mathbf{w}_{g}^p, \mathbf{w}_{g}^c, {\bm \rho}_g^p,{\bm \rho}_g^c,\mathbf{u}_g^p,\mathbf{u}_g^c|\,\, \forall g \in \mathcal{G}\right\rbrace$ is the set of optimization variables. Note, that problem \eqref{eq:Opt2} is the epigraph form of problem \eqref{eq:OptP}. According to \cite[Chapter 4]{convexOpt}, problems \eqref{eq:Opt2} and \eqref{eq:OptP} and their respective optimal solutions are equivalent. Therefore, we are able to use problem \eqref{eq:OptP} as an equivalent formulation of problem \eqref{eq:Opt2} throughout this proof for simplicity reasons. \\
\indent As Algorithm \ref{alg} is a block coordinate ascent algorithm operating iteratively, in iteration $\nu$, we solve the following convex optimization problem
\begin{subequations}\label{eq:OptPnu}
	\begin{align}
		\underset{\mathcal{V}_4}{\text{max}}\quad &\sum_{g\in\mathcal{G}} G_g(\mathbf{w}_{g}^p, \mathbf{w}_{g}^c, ({\bm \rho}_g^p)^\nu,({\bm \rho}_g^c)^\nu,(\mathbf{u}_g^p)^\nu,(\mathbf{u}_g^c)^\nu)  \\
		\text{s.t.} \quad & \eqref{eq:pmax1}, \nonumber \\
		& \sum_{{g \in \mathcal{G}_n^p}}G_g^p(\mathbf{w}_g^p,(\mathbf{u}_g^p)^\nu,({\bm \rho}_g^p)^\nu)  + \sum_{{g \in \mathcal{G}_n^c}}G_g^c(\mathbf{w}_g^c,(\mathbf{u}_g^c)^\nu,({\bm \rho}_g^c)^\nu) \leq F_{n}, &&\forall n \in \mathcal{N}, \label{eq:fnP}
	\end{align}
\end{subequations}  
where $\mathcal{V}_4 \triangleq \left\lbrace \mathbf{w}_{g}^p, \mathbf{w}_{g}^c|\,\, \forall g \in \mathcal{G}\right\rbrace$.
Note that the fixed values, i.e., $({\bm \rho}_g^p)^\nu,({\bm \rho}_g^c)^\nu,(\mathbf{u}_g^p)^\nu,(\mathbf{u}_g^c)^\nu$, are computed by $(\rho_{g,k}^p)^\nu = 1/e_{g,k,\text{mmse}}^p$, $(\rho_{i,k}^c)^\nu = 1/e_{i,k,\text{mmse}}^c$, \eqref{eq:upmmse}, and \eqref{eq:ucmmse}, using the optimal beamforming vectors computed in iteration $(\nu-1)$. We define the objective function of problem \eqref{eq:OptP} as $Q(\mathbf{w}_{g}^p, \mathbf{w}_{g}^c, {\bm \rho}_g^p,{\bm \rho}_g^c,\mathbf{u}_g^p,\mathbf{u}_g^c)$. Since $Q$ is a concave function and the achievable ergodic rates are bounded, the sequence $\{Q((\mathbf{w}_{g}^p)^\nu, (\mathbf{w}_{g}^c)^\nu, ({\bm \rho}_g^p)^\nu,({\bm \rho}_g^c)^\nu,(\mathbf{u}_g^p)^\nu,(\mathbf{u}_g^c)^\nu)\}_{\nu=0}^{\infty}$ increases monotonically after each iteration and converges to the limit point $\bar{Q}$. Since the feasible set defined by constraints \eqref{eq:pmax1} and \eqref{eq:fnP} is compact, $\{(\mathbf{w}_{g}^p)^\nu, (\mathbf{w}_{g}^c)^\nu\}_{\nu=0}^{\infty}$ must have a cluster point $\{\bar{\mathbf{w}}_{g}^p, \bar{\mathbf{w}}_{g}^c\}$. Meaning, it exists a subsequence $\{(\mathbf{w}_{g}^p)^{\nu_1}, (\mathbf{w}_{g}^c)^{\nu_1}\}_{\nu_1=\Lambda}^{\infty}$ for $\Lambda>0$ that converges to $\{\bar{\mathbf{w}}_{g}^p, \bar{\mathbf{w}}_{g}^c\}$. Thus, the following statement holds
\begin{equation}
	\underset{\nu_1\rightarrow\infty}{\text{lim}} \{ (\mathbf{w}_{g}^p)^{\nu_1}, (\mathbf{w}_{g}^c)^{\nu_1}, ({\bm \rho}_g^p)^{\nu_1},({\bm \rho}_g^c)^{\nu_1},(\mathbf{u}_g^p)^{\nu_1},(\mathbf{u}_g^c)^{\nu_1} \} = \{ \bar{\mathbf{w}}_{g}^p, \bar{\mathbf{w}}_{g}^c, \bar{\bm \rho}_g^p, \bar{\bm \rho}_g^c , \bar{\mathbf{u}}_g^p , \bar{\mathbf{u}}_g^c  \}.
\end{equation}
Note that $\bar{\bm \rho}_g^p, \bar{\bm \rho}_g^c , \bar{\mathbf{u}}_g^p$ , and $\bar{\mathbf{u}}_g^c$ are computed based on the beamforming vectors using the continuous functions \eqref{eq:upmmse}, \eqref{eq:ucmmse}, $\rho_{g,k}^p = 1/e_{g,k,\text{mmse}}^p$, and $\rho_{i,k}^c = 1/e_{i,k,\text{mmse}}^c$. At this point, we have shown that $\{ \bar{\bm \rho}_g^p, \bar{\bm \rho}_g^c , \bar{\mathbf{u}}_g^p, \bar{\mathbf{u}}_g^c \}$ are optimal when $\{\mathbf{w}_{g}^p, \mathbf{w}_{g}^c\} = \{\bar{\mathbf{w}}_{g}^p, \bar{\mathbf{w}}_{g}^c\}$. Following up, we show that the same applies vice versa, i.e., $\{\bar{\mathbf{w}}_{g}^p, \bar{\mathbf{w}}_{g}^c\}$ is optimal when $\{ {\bm \rho}_g^p,{\bm \rho}_g^c,\mathbf{u}_g^p,\mathbf{u}_g^c \} = \{ \bar{\bm \rho}_g^p, \bar{\bm \rho}_g^c , \bar{\mathbf{u}}_g^p , \bar{\mathbf{u}}_g^c \}$. With the optimal beamforming vectors from the previous iteration, and monotonicity of the objective function we write
\begin{align}
	Q((\mathbf{w}_{g}^p)^{\nu_1+1},& (\mathbf{w}_{g}^c)^{\nu_1+1}, ({\bm \rho}_g^p)^{\nu_1+1},({\bm \rho}_g^c)^{\nu_1+1},(\mathbf{u}_g^p)^{\nu_1+1},(\mathbf{u}_g^c)^{\nu_1+1}) \nonumber \\
	&\geq Q((\mathbf{w}_{g}^p)^{\nu_1+1}, (\mathbf{w}_{g}^c)^{\nu_1+1}, ({\bm \rho}_g^p)^{\nu_1},({\bm \rho}_g^c)^{\nu_1},(\mathbf{u}_g^p)^{\nu_1},(\mathbf{u}_g^c)^{\nu_1}) \nonumber \\
	&\geq Q(\mathbf{w}_{g}^p, \mathbf{w}_{g}^c, ({\bm \rho}_g^p)^{\nu_1},({\bm \rho}_g^c)^{\nu_1},(\mathbf{u}_g^p)^{\nu_1},(\mathbf{u}_g^c)^{\nu_1}), && \forall \mathbf{w}_{g}^p, \forall \mathbf{w}_{g}^c.
\end{align}
Taking the limit in this equation, we obtain the following relation
\begin{equation}
	\bar{Q} = Q( \bar{\mathbf{w}}_{g}^p, \bar{\mathbf{w}}_{g}^c, \bar{\bm \rho}_g^p, \bar{\bm \rho}_g^c , \bar{\mathbf{u}}_g^p , \bar{\mathbf{u}}_g^c ) \geq Q( \mathbf{w}_{g}^p, \mathbf{w}_{g}^c, \bar{\bm \rho}_g^p, \bar{\bm \rho}_g^c , \bar{\mathbf{u}}_g^p , \bar{\mathbf{u}}_g^c ), \qquad\; \forall \mathbf{w}_{g}^p, \forall \mathbf{w}_{g}^c.
\end{equation}
Therefore, $\{\bar{\mathbf{w}}_{g}^p, \bar{\mathbf{w}}_{g}^c\}$ are the optimal beamforming vectors of problem \eqref{eq:OptP}.
We have now shown that $\{\bar{\mathbf{w}}_{g}^p, \bar{\mathbf{w}}_{g}^c\}$ is optimal when $\{ {\bm \rho}_g^p,{\bm \rho}_g^c,\mathbf{u}_g^p,\mathbf{u}_g^c \} = \{ \bar{\bm \rho}_g^p, \bar{\bm \rho}_g^c , \bar{\mathbf{u}}_g^p , \bar{\mathbf{u}}_g^c \}$. At last, it can be shown that $\{ \bar{\mathbf{w}}_{g}^p, \bar{\mathbf{w}}_{g}^c, \bar{\bm \rho}_g^p, \bar{\bm \rho}_g^c , \bar{\mathbf{u}}_g^p , \bar{\mathbf{u}}_g^c \}$ is a KKT solution to problem \eqref{eq:OptP} by checking the KKT conditions. At this point, we have shown the solution set generated by Algorithm \ref{alg} converges to a KKT solution of problem \eqref{eq:OptP}. KKT points are not necessarily unique, however, any sequence $\{ (\mathbf{w}_{g}^p)^{\nu}, (\mathbf{w}_{g}^c)^{\nu}, ({\bm \rho}_g^p)^{\nu},({\bm \rho}_g^c)^{\nu},(\mathbf{u}_g^p)^{\nu},(\mathbf{u}_g^c)^{\nu} \}_{\nu=0}^{\infty}$ converges to the KKT solution in the limit. This proof relates on the equivalence of problem \eqref{eq:Opt2} and problem \eqref{eq:OptP}. Therefore, we conclude that $\{ \bar{\mathbf{w}}_{g}^p, \bar{\mathbf{w}}_{g}^c, \bar{\bm \rho}_g^p, \bar{\bm \rho}_g^c , \bar{\mathbf{u}}_g^p , \bar{\mathbf{u}}_g^c, \bar{R}_g^p, \bar{R}_g^c \}$ is also a KKT solution to problem \eqref{eq:OptP}, where $Q(\mathbf{w}_{g}^p, \mathbf{w}_{g}^c, {\bm \rho}_g^p,{\bm \rho}_g^c,\mathbf{u}_g^p,\mathbf{u}_g^c) = \sum_{g\in\mathcal{G}} (\bar{R}_g^p + \bar{R}_g^c)$. This completes the proof.
%\bibliographystyle{IEEEtran}
%\bibliography{IEEEabrv,reference}
%\balance
\bibliographystyle{IEEEtran}
\bibliography{bibliography}
\balance
\end{document}